\documentclass[twoside]{article}

\usepackage[accepted]{aistats2016}
\usepackage[usenames, dvipsnames, table]{xcolor}
\usepackage{tabularx, booktabs, caption, array}
\usepackage{multirow}

\usepackage{afterpage}

\usepackage{cellspace}
\usepackage{colortbl}
\usepackage{etoolbox}

\usepackage{lipsum}
\usepackage{graphicx} 
\usepackage{subfigure} 
\usepackage{algorithmic, algorithm}
\usepackage{hyperref}
\usepackage{booktabs}

\usepackage{multibib}
\newcites{appendix}{Appendix References}
\usepackage{amssymb, amsthm, amsfonts}
\usepackage[normalem]{ulem} 

\usepackage{enumerate}
\usepackage{multirow}
\usepackage{setspace}
\usepackage{array}
\usepackage{graphicx}
\usepackage{float}
\usepackage{standalone}
\usepackage{flushend}
\usepackage{xfrac}
\usepackage{nicefrac}
\usepackage{paralist}
\usepackage{enumitem}
\usepackage{xspace}
\usepackage{mathtools}  
\usepackage{relsize}
\usepackage{bbm}

\usepackage{balance} 

\usepackage{tikz}
\usetikzlibrary{arrows}
\usetikzlibrary{shapes}
\usetikzlibrary{shapes.misc}
\usetikzlibrary{fit}
\usetikzlibrary{calc}
\usetikzlibrary{intersections}
\usetikzlibrary{positioning}
\usetikzlibrary{decorations.pathmorphing}
\usetikzlibrary{decorations.pathreplacing}

\newcommand{\X}{\mathbf{X}}
\newcommand{\Y}{\mathbf{X}}

\newcommand{\A}{\mathbf{A}}

\newcommand{\Yhat}{\widehat{\Y}}

\newcommand{\Xtilde}{\widetilde{\X}}
\newcommand{\Ytilde}{\widetilde{\Y}}

\def\A{\mathbf{A}}
\def\Y{\mathbf{Y}}

\newcommand{\aprxrank}{r}

\newcommand{\transpose}{{\top}} 
\newcommand{\frob}{{\textnormal{F}}}
\newcommand{\supp}{\text{supp}}

\newcommand{\Tr}{\text{\textsc{Tr}}\mathopen{}}
\newcommand{\trace}{\Tr}

\newcommand{\OPT}{{\normalfont{\textsf{{\mbox{OPT}}}}}\xspace}

\newcommand{\kexp}{\mathsmaller{(k)}}

\newtheorem{theorem}{Theorem}

\newtheorem{lemma}{Lemma}

\newtheorem{fact}{Fact}
\newtheorem{corollary}{Corollary}

\newcommand{\eqdef}{{\triangleq}}

\usepackage[framemethod=tikz]{mdframed}
\definecolor{notecolor}{rgb}{0.122, 0.435, 0.698}
\newcommand{\mdfLABEL}[1]{%
	\node[align=left,anchor=north east, text width=3.2em,%
               outer sep=0pt,inner sep=0pt] at ($(O|-P) 
               -(0, 0.5em)$)
               {\textcolor{notecolor}{\textsc{Note}}};}
                 
\mdfdefinestyle{noteframe}{topline=false,rightline=false,bottomline=true,%
innerleftmargin=1em,innertopmargin=0.5em, innerbottommargin=1em, innerrightmargin=0em,
skipabove=1em,linecolor=white,rightmargin=0em,skipbelow=1em,%
    tikzsetting={draw=notecolor,line width=2pt},%
    firstextra={\mdfLABEL},%
    singleextra={\mdfLABEL},%
    secondextra={\mdfLABEL},%
    middleextra={\mdfLABEL},%
}
\newif\ifhidenotes

\makeatletter
\newtheorem*{rep@theorem}{\rep@title}
\newcommand{\newreptheorem}[2]{%
\newenvironment{rep#1}[1]{%
 \def\rep@title{#2 \ref{##1}}%
 \begin{rep@theorem}}%
 {\end{rep@theorem}}}
\makeatother
\newreptheorem{theorem}{Theorem}
\newreptheorem{lemma}{Lemma}

\DeclareMathOperator*{\argmax}{arg\,max}
\DeclareMathOperator*{\argmin}{arg\,min}

\newcommand{\problemname}[1]{{\normalfont{\textsf{{\mbox{\textsc{#1}}}}}}\xspace}

\newcommand{\bcc}{\problemname{BCC}}
\newcommand{\cc}{\problemname{CC}}
\newcommand{\maxagree}{\problemname{MaxAgree}}
\newcommand{\mindisagree}{\problemname{MinDisagree}}
\newcommand{\agree}{{\normalfont{\textsf{\small{\mbox{\textsc{Agree}}}}}}\xspace}

\newcommand{\minus}{\text{$-$}\xspace}
\newcommand{\plus}{\text{$+$}\xspace}
\newcommand{\Tsvd}{\mathrm{T}_{\mathsmaller{\mathsf{SVD}}}\xspace}
\newcommand{\Tx}{\mathrm{T}_{\mathcal{X}}}
\newcommand{\Ty}{\mathrm{T}_{\mathcal{Y}}}

\usepackage{stfloats}

\newcommand{\lowrankalgoname}{\texttt{BiLinearLowRankSolver}}

\begin{document}

%
\runningauthor{
  Asteris,
  Kyrillidis,
  Papailiopoulos,
  Dimakis
}

\twocolumn[

\aistatstitle{Bipartite Correlation Clustering -- Maximizing Agreements}

\vspace{-.8em}
\aistatsauthor{%
Megasthenis Asteris \And
Anastasios Kyrillidis}%
\aistatsaddress{%
The University of Texas at Austin \And
The University of Texas at Austin}%
\vspace{-.8em}
\aistatsauthor{%
Dimitris Papailiopoulos \And
Alexandros G. Dimakis
}%
\aistatsaddress{%
University of California, Berkeley \And
The University of Texas at Austin
} 

%

] 

\begin{abstract}
In \emph{Bipartite Correlation Clustering} (\bcc) 
we are given a complete \emph{bipartite} graph $G$ with `\plus' and `\minus' edges,
and we seek a vertex clustering that maximizes the number of \emph{agreements}:
the number of all `\plus' edges within clusters plus all `\minus' edges cut across clusters.
\bcc is known to be NP-hard~\cite{amit2004bicluster}.

We present a novel approximation algorithm for $k$-\bcc, a variant of \bcc with an upper bound $k$ on the number of clusters.
Our algorithm outputs a $k$-clustering that provably achieves a number of agreements within a multiplicative ${(1-\delta)}$-factor from the optimal, for any desired accuracy $\delta$.
It relies on solving a combinatorially constrained bilinear maximization on the bi-adjacency matrix of $G$.
It runs in time exponential in~$k$ and~$\sfrac{1}{\delta}$, but linear in the size of the input.

Further, we show that, in the (unconstrained) \bcc setting,
an ${(1-\delta)}$-approximation can be achieved by
$O(\delta^{-1})$ clusters regardless of the size of the graph.
In turn, our $k$-\bcc algorithm implies an Efficient PTAS for the \bcc objective of maximizing agreements.

\end{abstract}


\section{Introduction}

Correlation Clustering (\cc)~\cite{bansal2004correlation} considers the task of partitioning a set of objects into clusters based on their pairwise relationships.
It arises naturally in several areas such as in network monitoring~\cite{ahn2015correlation}, document clustering~\cite{bansal2004correlation} and textual similarity for data integration~\cite{cohen2002learning}. 
In its simplest form, objects are represented as vertices in a {complete} graph whose edges are labeled `\plus' or `\minus' to encode \emph{similarity} or \emph{dissimilarity} among vertices, respectively.
The objective is to compute a vertex partition that maximizes the number of \emph{agreements} with the underlying graph,
\textit{i.e.}, the total number of `\plus' edges in the interior of clusters plus the number of `\minus' edges across clusters.
The number of output clusters is itself an optimization variable --{not} part of the input.
It may be meaningful, however, to restrict the number of output clusters to be at most $k$.
The constrained version is known as $k$-\cc
and similarly to the unconstrained problem, it is NP-hard~\cite{bansal2004correlation,giotis2006correlation,karpinski2009linear}.
A significant volume of work has focused on approximately solving the problem of maximizing the number of agreements (\maxagree), or the equivalent --yet more challenging in terms of approximation-- objective of minimizing the number of \emph{disagreements} (\mindisagree)~\cite{bansal2004correlation, demaine2006correlation, charikar2003clustering, ailon2008aggregating, van2009deterministic, ailon2009correlation}.


Bipartite Correlation Clustering (\bcc)
is a natural variant of \cc on bipartite graphs.
Given a complete {bipartite} graph $G=(U,V, E)$ with edges labeled `\plus' or `\minus',  the objective is once again to compute a clustering of the vertices that maximizes the number of agreements with the labeled pairs.
\bcc is a special case of the \emph{incomplete} \cc problem~\cite{charikar2003clustering}, where only a subset of vertices are connected and non-adjacent vertices do not affect the objective function.
The output clusters may contain vertices from either one or both sides of the graph.
Finally, we can define the $k$-\bcc variant that enforces an upper bound on the number of output clusters, similar to $k$-\cc for \cc.

The task of clustering the vertices of a bipartite graph is common across many areas of machine learning.
Applications include recommendation systems~\cite{symeonidis2007nearest, vlachos2014improving}, where analyzing the structure of a large sets of pairwise interactions (\textit{e.g.}, among users and products) allows useful predictions about future interactions,
gene expression data analysis \cite{amit2004bicluster, madeira2004biclustering} and graph partitioning problems in data mining \cite{fern2004solving, zha2001bipartite}.

Despite the practical interest in bipartite graphs,
there is limited work on \bcc
and it is focused on the theoretically more challenging \mindisagree objective:
\cite{amit2004bicluster}
established an $11$-approximation algorithm, while~\cite{ailon2012improved} achieved a $4$-approximation, the currently best known guarantee. 
Algorithms for incomplete \cc \cite{demaine2006correlation, charikar2003clustering, swamy2004correlation} can be applied to \bcc,
but they do not leverage the structure of the bipartite graph.
Moreover, existing approaches for incomplete \cc rely on LP or SDP solvers which scale poorly.

\paragraph{Our contributions}
We develop a novel approximation algorithm for $k$-\bcc 
with provable guarantees for the \maxagree objective.
Further, we show that under an appropriate configuration, our algorithm yields an Efficient Polynomial Time Approximation Scheme (EPTAS%
\footnote{
EPTAS
refers to an algorithm that approximates the solution of an optimization problem within a multiplicative $(1-\epsilon)$-factor, for any constant $\epsilon \in (0,1)$, and has complexity that scales arbitrarily in $1/\epsilon$, but as a constant order polynomial (independent of $\epsilon$) in the input size~$n$.
EPTAS is more efficient than a PTAS; 
for example, a running time of $O(n^{1/\epsilon})$ is considered a PTAS, but not an EPTAS.
}) %
for the unconstrained \bcc problem.
Our contributions can be summarized as follows:
\begin{enumerate}[leftmargin=*, itemsep=0.2em, topsep=0.2em]
\item
\emph{$k$-\bcc}:
Given a bipartite graph $G=(U,V,E)$, a parameter $k$, and any constant accuracy parameter ${\delta \in (0, 1)}$,
our algorithm computes a clustering of $U \cup V$ into at most $k$ clusters and achieves a number of agreements that lies within a ${(1-\delta)}$-factor from the optimal.
It runs in time exponential in~$k$ and~$\delta^{-1}$, but linear in the size of~$G$.
\item
\emph{\bcc}:
In the unconstrained \bcc setting, 
the optimal number of clusters may be anywhere from~$1$ to $|U|+|V|$.
We show that if one is willing to settle for a ${(1-\delta)}$-approximation of the \maxagree objective,
it suffices to use at most $O(\delta^{-1})$ clusters, regardless of the size of $G$.
In turn, under an appropriate configuration, our $k$-\bcc algorithm yields an EPTAS for the (unconstrained) \bcc problem.
\item
Our algorithm relies on formulating the $k$-\bcc/ \maxagree problem as a combinatorially constrained bilinear maximization
\vspace{-2pt}
\begin{align}
	\max_{\mathbf{X} \in \mathcal{X}, \mathbf{Y} \in \mathcal{Y}} \quad
	&   
	\trace\left(\mathbf{X}^\transpose \mathbf{B} \mathbf{Y} \right),
	\label{intro:bilinear}
\end{align}
where $\mathbf{B}$ is the bi-adjacency matrix of~$G$,
and $\mathcal{X}$, $\mathcal{Y}$ are the sets of cluster assignment matrices for $U$ and $V$, respectively.
In Sec.~\ref{sec:bilmax}, we briefly describe our approach for approximately solving~\eqref{intro:bilinear} and its guarantees under more general constraints.
%
\end{enumerate}
We note that our $k$-\bcc algorithm and its guarantees can be extended to \emph{incomplete}, but dense \bcc instances, where the input~$G$ is \emph{not} a complete bipartite graph, but $|E| = \Omega(|U|\cdot|V|)$
For simplicity, we restrict the description to the complete case.
Finally, we supplement our theoretical findings with experimental results on synthetic and real datasets.


\subsection{Related work}

There is extensive literature on \cc;
see~\cite{bonchi2014correlation} for a list of references.
Here, we focus on bipartite variants. 

\bcc was introduced by Amit in~\cite{amit2004bicluster}
along with an $11$-approximation algorithm for the \mindisagree objective,
based on a linear program (LP) with $O(|E|^{3})$ constraints.
In~\cite{ailon2012improved},
Ailon et al. proposed two algorithms for the same 
objective: 
\textit{i}) a deterministic $4$-approximation based on an LP formulation and de-randomization arguments by~\cite{van2009deterministic}, and 
\textit{ii}) a randomized $4$-approximation combinatorial algorithm.
The latter is computationally more efficient with complexity scaling linearly in the size of the input, compared to the
LP that has $O\bigl(\left(|V| + |U|\right)^3\bigr)$ constraints.

For the \emph{incomplete} \cc problem, which encompasses \bcc as a special case, \cite{demaine2006correlation} provided an LP-based $O\left(\log n \right)$-approximation for \mindisagree.
A similar result was achieved by~\cite{charikar2003clustering}. 
For the \maxagree objective, ~\cite{ailon2012improved, swamy2004correlation} proposed an SDP relaxation, similar to that for MAX $k$-CUT, and achieved a $0.7666$-approximation.
We are not aware of any results explicitly on the \maxagree objective for either \bcc or $k$-\bcc.
For comparison, in the $k$-\cc setting~\cite{bansal2004correlation} provided a simple $3$-approximation algorithm for ${k = 2}$,
while for ${k \ge 2}$ \cite{giotis2006correlation} provided a PTAS for \maxagree
and one for \mindisagree.
\cite{karpinski2009linear} improved on the latter utilizing approximations schemes to the Gale-Berlekamp switching game.
Table~\ref{table:algo_comp_summary} summarizes the aforementioned results.

Finally, we note that our algorithm relies on adapting ideas from~\cite{asteris2015sparseviabipartite} for approximately solving a combinatorially constrained quadratic maximization.

\begin{table*}[t!]
\centering
\small
\rowcolors{2}{white}{black!05!white}
\begin{tabular}{ccccccccccc}
  \toprule
  Ref. & & Min./Max. & & Guarantee & & Complexity & & D/R & & Setting \\ 
  \cmidrule{1-1} \cmidrule{3-3} \cmidrule{5-5} \cmidrule{7-7} \cmidrule{9-9} \cmidrule{11-11} 
    \cite{giotis2006correlation} & & \maxagree & & (E)PTAS & & $\sfrac{n}{\delta} \cdot k^{O\left(\delta^{-2} \log k \log (\sfrac{1}{\delta}) \right)}$ & & R  & & $k$-\cc \\ 
     \cite{giotis2006correlation} & & \mindisagree & & PTAS & & $n^{O(\sfrac{100^k}{\delta^2})} \cdot \log n$ & & R  & & $k$-\cc \\ 
     \cite{karpinski2009linear} & & \mindisagree & & PTAS & & $n^{O(\sfrac{9^k}{\delta^2})} \cdot \log n$ & & R  & & $k$-\cc \\          
     \cite{demaine2006correlation, charikar2003clustering} & & \mindisagree & & $O(\log n)$-\OPT & & LP & & D & & Inc. \cc \\
     \cite{swamy2004correlation} & & \maxagree & & $0.766$-\OPT & & SDP & & D & & Inc. \cc \\
     \cite{amit2004bicluster} & & \mindisagree & & 11-\OPT & & LP & & D & & \bcc \\
     \cite{ailon2012improved} & & \mindisagree & & 4-\OPT & & LP & & D & & \bcc \\
       \cite{ailon2012improved} & & \mindisagree & & 4-\OPT & & $|E|$ & & R & & \bcc \\
  \midrule
  \rowcolors{2}{white}{white}
       Ours & & \maxagree & & (E)PTAS & &
       $2^{O(\sfrac{k}{\delta^{2}} \cdot \log{\sfrac{\sqrt{k}}{\delta}})} \cdot (\delta^{-2}+k) \cdot {n} + \Tsvd(\delta^{-2})$
       & & R & & $k$-\bcc \\      
      & &  \maxagree & & (E)PTAS & & 
      $2^{O(\delta^{-3} \cdot \log{{\delta^{-3}}})} \cdot O(\delta^{-2}) \cdot {n} + \Tsvd(\delta^{-2})$
      & & R & & \bcc \\
  \bottomrule
\end{tabular}
\caption{Summary of results on \bcc and related problems.
For each scheme, we indicate the problem setting, 
the objective (\maxagree/\mindisagree),
the guarantees ($c$-{\OPT} implies a multiplicative factor approximation
),
and its computational complexity ($n$ denotes the total number of vertices, LP/SDP denotes the complexity of a linear/semidefinite program, and $\Tsvd(r)$ the time required to compute a rank-$r$ truncated SVD of a $n \times n$ matrix).
The D/R column indicates whether the scheme is deterministic or randomized.} 
\label{table:algo_comp_summary}
\end{table*}

\section{$k$-\bcc as a Bilinear Maximization}
\label{sec:bilinear}

We describe the $k$-\bcc problem
and show that it can be formulated as a combinatorially constrained bilinear maximization on the bi-adjacency matrix of the input graph.
Our $k$-\bcc algorithm relies on approximately solving that bilinear maximization.

\textbf{Maximizing Agreements}\;
An instance of \mbox{$k$-\bcc}
consists of an undirected, complete, bipartite graph ${G=(U,V, E)}$
whose edges have binary $\pm 1$ weights,
and a parameter $k$.
The objective, referred to as $\maxagree[k]$, is to compute a clustering of the vertices into at most $k$ clusters, such that the total number of positive edges in the interior of clusters plus the number of negative edges across clusters is maximized.

Let $E^{+}$ and $E^{-}$ denote the sets of positive and negative edges, respectively.
Also, let ${m=|U|}$ and ${n =|V|}$.
The bipartite graph $G$ can be represented by its weighted bi-adjacency matrix $
	\mathbf{B}
	\in \lbrace \pm 1 \rbrace^{m \times n}
$,
where
$B_{ij}$ is equal to the weight of edge $(i,j)$.

Consider a clustering $\mathcal{C}$ of the vertices $U \cup V$ into \emph{at most} $k$ clusters $C_{1}, \hdots, C_{k}$.
Let $\mathbf{X} \in \lbrace 0,1\rbrace^{m \times k}$ be the cluster assignment matrix for the vertices of $U$ associated with $\mathcal{C}$,
that is, $X_{ij}=1$ if and only if vertex $i \in U$ is assigned to cluster $C_{j}$.
Similarly, let $\mathbf{Y} \in \lbrace 0,1\rbrace^{n \times k}$ be the cluster assignment for $V$.

\begin{lemma}
\label{lemma:num-agreements}
For any instance ${G=(U, V, E)}$ of the $k$-\bcc problem with bi-adjacency matrix $\mathbf{B} \in \lbrace \pm 1 \rbrace^{|U| \times |V|}$,
and for any clustering $\mathcal{C}$ of $U \cup V$ into $k$ clusters,
the number of agreements achieved by $\mathcal{C}$ is 
\begin{align}
	\agree\bigl( \mathcal{C}\bigr)
	=
	\trace\bigl(\mathbf{X}^{\transpose}\mathbf{B}\mathbf{Y}\bigr) + \lvert E^{\minus} \rvert,
	\nonumber
\end{align}
where ${\mathbf{X}\in \lbrace 0,1\rbrace^{|U| \times k}}$ and ${\mathbf{Y}\in \lbrace 0,1\rbrace^{|V| \times k}}$ are the cluster assignment matrices corresponding to $\mathcal{C}$.
\end{lemma}
\begin{proof}
Let $\mathbf{B}^{+}$ and $\mathbf{B}^{-} \in \lbrace 0, 1\rbrace^{|U| \times |V|}$ 
be indicator matrices for $E^{+}$ and $E^{-}$.
Then, ${\mathbf{B} = \mathbf{B}^{+} - \mathbf{B}^{-}}$.
Given a clustering ${\mathcal{C} = \lbrace C_{j} \rbrace_{j=1}^{k}}$,
or equivalently the assignment matrices  $\mathbf{X} \in \lbrace 0, 1\rbrace^{|U| \times k}$ and $\mathbf{Y} \in \lbrace 0, 1\rbrace^{|V| \times k}$,
the number of pairs of similar vertices assigned to the same cluster is equal to 
$
	  \trace\left(\mathbf{X}^{\transpose}\mathbf{B}^{+}\mathbf{Y}\right).
$
Similarly, 
the number of pairs of dissimilar vertices assigned to different clusters is equal to 
$
	  \lvert E^{\minus} \rvert
	  -
	  \trace\left(\mathbf{X}^{\transpose}\mathbf{B}^{-}\mathbf{Y}\right).
$
The total number of agreements achieved by $\mathcal{C}$
is
\begin{align}
  \agree\bigl( \mathcal{C}\bigr)
  &=
  \trace\bigl(\mathbf{X}^{\transpose}\mathbf{B}^{+}\mathbf{Y}\bigr)
  +
  \lvert E^{\minus} \rvert
  -
  \trace\bigl(\mathbf{X}^{\transpose}\mathbf{B}^{-}\mathbf{Y}\bigr) 
  \nonumber\\
  &=
  \trace\bigl(
  \mathbf{X}^{\transpose}\mathbf{B}\mathbf{Y}
	\bigr)
	+ 
	\lvert E^{\minus} \rvert,
  \nonumber
\end{align}
which is the desired result.
\end{proof}

It follows that computing a $k$-clustering that achieves the maximum number of agreements,
boils down to a constrained bilinear maximization:
\begin{align}
   \maxagree[k]
   =
   \max_{
		 \substack{
			\mathbf{X} \in \mathcal{X}, \mathbf{Y} \in \mathcal{Y}
		 }
   }
   \trace\bigl(\mathbf{X}^{\transpose}\mathbf{B}\mathbf{Y}\bigr) + \lvert E^{\minus} \rvert,
   \label{bcc:maxagree}
\end{align}
{where 
\setlength\abovedisplayskip{0pt plus 2pt minus 2pt}
\begin{equation}
\begin{aligned}
  &\mathcal{X}
  \eqdef
  \bigl\lbrace \mathbf{X}\in \lbrace 0, 1\rbrace^{m \times k}: \|\mathbf{X}\|_{\infty, 1} = 1\bigr\rbrace,\\
 &\mathcal{Y}
  \eqdef
  \bigl\lbrace \mathbf{Y}\in \lbrace 0, 1\rbrace^{n \times k}: \|\mathbf{Y}\|_{\infty, 1} = 1\bigr\rbrace.
\end{aligned}
\label{cluster-assignment-sets}
\end{equation}}%
Here, $\|\mathbf{X}\|_{\infty, 1}$ denotes the maximum of the $\ell_{1}$-norm of the rows of~$\mathbf{X}$.
Since ${\mathbf{X} \in \lbrace 0, 1 \rbrace^{|U|\times k}}$, the constraint ensures that each row of $\mathbf{X}$ has exactly one nonzero entry.
Hence, $\mathcal{X}$ and $\mathcal{Y}$ describe the sets of valid cluster assignment matrices for~$U$ and~$V$, respectively.

In the sequel, we briefly describe our approach for approximately solving the maximization in~\eqref{bcc:maxagree} under more general constraints, and
subsequently apply it to the $k$-\bcc/\maxagree problem.


\section{Bilinear Maximization Framework}
\label{sec:bilmax}

We describe a simple algorithm for computing an approximate solution to the constrained maximization
\begin{align}
   \max_{\mathbf{X} \in \mathcal{X}, \mathbf{Y} \in \mathcal{Y}} \quad
   &   
   \trace\left(\X^\transpose \A \mathbf{Y} \right),
   \label{bilinear-abstract}
\end{align}
where the input argument $\A$ is a real $m \times n$ matrix,
and $\mathcal{X}$,  $\mathcal{Y}$ are norm-bounded sets.
Our approach is exponential in the rank of the argument~$\mathbf{A}$,
which can be prohibitive in practice.
To mitigate this effect, we solve the maximization on a low-rank approximation~$\widetilde{\A}$ of~$\A$, instead.
The quality of the output depends on the spectrum of~$\mathbf{A}$ and the rank~$r$ of the surrogate matrix.
Alg.~\ref{algo:bilinear-lowrank} outlines our approach for solving~\eqref{bilinear-abstract}, operating directly on the rank-$r$ matrix~$\widetilde{\A}$.

If we knew the optimal value for variable $\mathbf{Y}$,
then the optimal value for variable $\mathbf{X}$
would be the solution to 
\begin{align}
   \mathrm{P}_{\mathcal{X}}(\mathbf{L}) = \argmax_{\mathbf{X} \in \mathcal{X}}
   \trace\left(\mathbf{X}^{\top} \mathbf{L}\right),
   \label{solve-for-X}
\end{align}
for $\mathbf{L} = \widetilde{\mathbf{A}}\mathbf{Y}$.
Similarly, 
with
$\mathbf{R} = \widetilde{\mathbf{A}}^{\transpose}\mathbf{X}$ for a given $\mathbf{X}$,
   \begin{align}
      \mathrm{P}_{\mathcal{Y}}(\mathbf{R}) = \argmax_{\mathbf{Y} \in \mathcal{Y}} \trace\bigl(\mathbf{R}^{\transpose}\mathbf{Y}\bigr)
   \label{solve-for-Y}
   \end{align}
is the optimal value of $\mathbf{Y}$ for that $\mathbf{X}$.
Our algorithm requires that such a linear maximization oracles $\mathrm{P}_{\mathcal{X}}(\cdot)$ and $\mathrm{P}_{\mathcal{Y}}(\cdot)$ exist.%
   \footnote{
   In the case of the $k$-\bcc problem~\eqref{bcc:maxagree}, where $\mathcal{X}$ and $\mathcal{Y}$ correspond to the sets of cluster assignment matrices defined in~\eqref{cluster-assignment-sets},
   such optimization oracles exist (see Alg.~\ref{algo:zero-one-orthogonal}).}
Of course, the optimal value for either of the two variables is not known.
It is known, however, that the columns of the $m \times k$ matrix $\mathbf{L}=\widetilde{\mathbf{A}}\mathbf{Y}$ lie in the $r$-dimensional range of $\widetilde{\mathbf{A}}$ for all feasible $\mathbf{Y} \in \mathcal{Y}$.
Alg.~\ref{algo:bilinear-lowrank} effectively operates by sampling candidate values for $\mathbf{L}$.
More specifically, it considers a large --exponential in ${{r} \cdot {k}}$-- collection of $r \times k$ matrices~$\mathbf{C}$
whose columns are points of an $\epsilon$-net of the unit $\ell_{2}$-ball $\mathbb{B}_{2}^{r-1}$.
Each matrix $\mathbf{C}$ is used to generate a matrix $\mathbf{L}$ whose $k$ columns lie in the range of the input matrix $\widetilde{\mathbf{A}}$
and in turn to produce a feasible solution pair $\mathbf{X}, \mathbf{Y}$ by successively solving~\eqref{solve-for-X} and~\eqref{solve-for-Y}.
Due to the properties of the $\epsilon$-net, one of the computed feasible pairs is guaranteed to achieve an objective value in~\eqref{bilinear-abstract} close to the optimal one (for argument $\widetilde{\mathbf{A}}$).

\begin{lemma}
   \label{low-rank-solver-guarantees}
   For any  real ${m \times n}$, rank-$r$ matrix~$\widetilde{\mathbf{A}}$,
   and sets
   ${\mathcal{X} \subset \mathbb{R}^{m \times k}}$ and
   ${\mathcal{Y} \subset \mathbb{R}^{n \times k}}$,
   let
   $$
   \bigl(\widetilde{\mathbf{X}}_{\star}, \widetilde{\mathbf{Y}}_{\star}\bigr)
   \eqdef
   \argmax_{\mathbf{X} \in \mathcal{X}, \mathbf{Y} \in \mathcal{Y}}\trace\bigl(\mathbf{X}^{\transpose}{\widetilde{\mathbf{A}}}\mathbf{Y}\bigr).
   $$
   If 
   there exist linear maximization oracles $\mathrm{P}_{\mathcal{X}}$ and $\mathrm{P}_{\mathcal{Y}}$ as in~\eqref{solve-for-X} and~\eqref{solve-for-Y},
   then 
   Alg.~\ref{algo:bilinear-lowrank} 
   with input $\widetilde{\mathbf{A}}$ and accuracy ${\epsilon \in (0,1)}$
   outputs $\widetilde{\mathbf{X}} \in \mathcal{X}$ 
   and $\widetilde{\mathbf{Y}} \in \mathcal{Y}$ such that
   \begin{align}
	  \trace\bigl( \widetilde{\mathbf{X}}^{\transpose} \widetilde{\mathbf{A}} \widetilde{\mathbf{Y}} \bigr)
	  \ge
	  \trace\bigl( \widetilde{\mathbf{X}}_{\star}^{\transpose} \widetilde{\mathbf{A}} \widetilde{\mathbf{Y}}_{\star} \bigr)
	  - 2  \epsilon  \sqrt{k} \cdot \|\widetilde{\mathbf{A}}\|_{2} \cdot \mu_{\mathcal{X}} \cdot \mu_{\mathcal{Y}},
	  \nonumber
   \end{align}
   where
   $\mu_{\mathcal{X}}\eqdef \max_{\mathbf{X} \in  \mathcal{X}}\| \mathbf{X} \|_{\frob}$
   and
   $\mu_{\mathcal{Y}}\eqdef \max_{\mathbf{Y} \in  \mathcal{Y}}\| \mathbf{Y} \|_{\frob}$,
   in time $O\mathopen{}\bigl(\bigl(\sfrac{2\sqrt{r}}{\epsilon}\bigr)^{r \cdot k} \cdot \bigl(\Tx+\Ty+(m+n)r\bigr)\bigr) + \Tsvd(r)$.
   Here, $\Tsvd(r)$ denotes the time required to compute the truncated  SVD of $\widetilde{\mathbf{A}}$, while $\Tx$ and $\Ty$ denote the running time of the two oracles.
\end{lemma}
\begin{algorithm}[t!]
   \caption{ {\lowrankalgoname} }
   \label{algo:bilinear-lowrank}
   {\color{black}   	\small
   \begin{algorithmic}[1]
   \INPUT:
   $m \times n$ real matrix $\widetilde{\mathbf{A}}$ of rank $r$,
   ${\epsilon \in (0,1)}$
   \OUTPUT $\widetilde{\mathbf{X}} \in \mathcal{X}$, $\widetilde{\mathbf{Y}} \in \mathcal{Y}$
   \hfill \COMMENT{See Lemma~\ref{low-rank-solver-guarantees}.}
   \STATE $\mathcal{C} \gets \lbrace \rbrace$ \hfill \COMMENT{Candidate solutions}
   \STATE 
   ${
	  \widetilde{\mathbf{U}}, \widetilde{\mathbf{\Sigma}}, \widetilde{\mathbf{V}}
	  \leftarrow
	  \texttt{SVD}(\widetilde{\mathbf{A}})
   }$
   \hfill \COMMENT{
      ${
         \widetilde{\mathbf{\Sigma}} \in \mathbb{R}^{\aprxrank \times \aprxrank}
      }$
   }
   \FOR{\textbf{each} $\mathbf{C} \in (\epsilon\text{-net of } \mathbb{B}_{2}^{r-1})^{\otimes k}$} 
      \STATE  $\mathbf{L} \gets \widetilde{\mathbf{U}}\widetilde{\mathbf{\Sigma}}\mathbf{C}$
      \hfill \COMMENT{$\mathbf{L} \in \mathbb{R}^{m \times k}$}
      \STATE ${\X} \gets$ $\mathrm{P}_{\mathcal{X}}(\mathbf{L})$
	  \STATE ${\mathbf{R}} \gets \X^{\transpose}\widetilde{\mathbf{A}}$ 
	  \hfill \COMMENT{$\mathbf{R} \in \mathbb{R}^{k \times n}$}
	  \STATE ${\Y} \gets$ $\mathrm{P}_{\mathcal{Y}}(\mathbf{R})$
      \STATE $\mathcal{C} \gets \mathcal{C} \cup \bigl\lbrace{({\X}, {\Y})}\bigr\rbrace$
   \ENDFOR
   \STATE $( \widetilde{\mathbf{X}}, \widetilde{\mathbf{Y}}) \gets
   \arg\max_{\substack{ (\mathbf{X}, \mathbf{Y}) \in \mathcal{C} } }
   \trace\bigl(\mathbf{X}^{\transpose}\widetilde{\mathbf{U}}\widetilde{\mathbf{\Sigma}}\widetilde{\mathbf{V}}^{\transpose}\mathbf{Y}\bigr)$
   \end{algorithmic}
   }
\end{algorithm}
The crux of Lemma~\ref{low-rank-solver-guarantees} is that if there exists an efficient procedure to solve the simpler maximizations~\eqref{solve-for-X} and~\eqref{solve-for-Y},
then we can approximately solve the bilinear maximization~\eqref{bilinear-abstract} in time that depends exponentially on the intrinsic dimension $r$ of the input~$\widetilde{\mathbf{A}}$, but polynomially on the size of the input.
A formal proof is given in Appendix Sec.~\ref{sec:approx-proof}.

Recall that $\widetilde{\mathbf{A}}$ is only a low-rank approximation of 
the potentially full rank original argument~$\mathbf{A}$.
The guarantees of Lemma~\ref{low-rank-solver-guarantees} can be translated to guarantees for the original matrix introducing an additional error term to account for the extra level of approximation.
\begin{lemma}
   \label{full-rank-guarantees}
   For any $\mathbf{A} \in \mathbb{R}^{m \times n}$, let
   $$
   \bigl( \mathbf{X}_{\star}, \mathbf{Y}_{\star} \bigr) 
   \eqdef 
   \argmax_{\mathbf{X}\in \mathcal{X}, \mathbf{Y}\in \mathcal{Y}}\trace\bigl(\mathbf{X}^{\transpose}\mathbf{A}\mathbf{Y}\bigr),
   $$
   where $\mathcal{X}$ and $\mathcal{Y}$ satisfy the conditions of Lemma~\ref{low-rank-solver-guarantees}.
   Let $\widetilde{\mathbf{A}}$ be a rank-$r$ approximation of $\mathbf{A}$,
   and $\widetilde{\mathbf{X}} \in \mathcal{X}$, $\widetilde{\mathbf{Y}} \in \mathcal{Y}$ 
   be the output of Alg.~\ref{algo:bilinear-lowrank}
   with input $\widetilde{\mathbf{A}}$ and accuracy $\epsilon$. Then,
   \begin{align}
      &\trace\bigl(\mathbf{X}_{\star}^{\transpose}\mathbf{A}\mathbf{Y}_{\star}\bigr)
      -
      \trace\bigl(\widetilde{\mathbf{X}}^{\transpose}\mathbf{A}\widetilde{\mathbf{Y}}\bigr)
      \nonumber\\&\quad \le
      2  \cdot \left(
         \epsilon  \sqrt{k} \cdot \|\widetilde{\mathbf{A}}\|_{2}
         +
         \|\mathbf{A}-{\widetilde{\mathbf{A}}}\|_{2} 
         \right)
     \cdot \mu_{\mathcal{X}}
     \cdot \mu_{\mathcal{Y}}.
     \nonumber
   \end{align}
\end{lemma}
Lemma~\ref{full-rank-guarantees} follows from Lemma~\ref{low-rank-solver-guarantees}.
A formal proof is deferred to Appendix Sec.~\ref{sec:approx-proof}.
This concludes the brief discussion on our bilinear maximization framework.

\begin{algorithm}[t!]
   \small
   \caption{
      \small{ 
      $\mathrm{P}_{\mathcal{X}}(\cdot)$,
      $\mathcal{X}
       \eqdef
       \bigl\lbrace \mathbf{X} \in \lbrace 0, 1\rbrace^{m \times k}: \|\mathbf{X}\|_{\infty, 1} = 1\bigr\rbrace
      $} 
   }
   \label{algo:zero-one-orthogonal}
   {\color{black}   
   \begin{algorithmic}[1]
   \INPUT:
   $m \times k$ real matrix $\mathbf{L}$
   \OUTPUT $\Xtilde = \argmax_{\mathbf{X} \in \mathcal{X}}\trace\bigl(\X^{\transpose}\mathbf{L}\bigr)$
   \STATE $\Xtilde \gets \mathbf{0}_{m \times k}$
   \FOR{$i=1,\hdots, m$ } 
      \STATE $j_{i} \gets \argmax_{j \in [k]} L_{ij}$ 
      \STATE $\widetilde{X}_{i j_{i}} \gets 1$
   \ENDFOR
   \end{algorithmic}
   }
\end{algorithm}

\section{Our $k$-BCC Algorithm}
The $k$-\bcc/\maxagree problem on a bipartite graph ${G=(U,V,E)}$ can be written as a constrained bilinear maximization~\eqref{bcc:maxagree} on the bi-adjacency matrix $\mathbf{B}$
over the sets of valid cluster assignment matrices~\eqref{cluster-assignment-sets} for $V$ and $U$.
Adopting the framework of the previous section to approximately solve the maximization in~\eqref{bcc:maxagree} we obtain our $k$-\bcc algorithm.

The missing ingredient is a pair of efficient procedures $\mathrm{P}_{\mathcal{X}}(\cdot)$ and $\mathrm{P}_{\mathcal{Y}}(\cdot)$,
as described in Lemma~\ref{low-rank-solver-guarantees} 
for solving~\eqref{solve-for-X} and~\eqref{solve-for-Y}, respectively,
in the case where $\mathcal{X}$ and $\mathcal{Y}$ are the sets of valid cluster assignment matrices~\eqref{cluster-assignment-sets}.
Such a procedure exists and is outlined in Alg.~\ref{algo:zero-one-orthogonal}.

\begin{lemma}
   \label{lemma:local-solver-for-bcc}
   For any $\mathbf{L} \in \mathbb{R}^{m \times k}$, Algorithm~\ref{algo:zero-one-orthogonal}
   outputs
   $$
      \argmax_{\mathbf{X}\in \mathcal{X}}\trace\bigl(\X^{\transpose}\mathbf{L}\bigr),
   $$
   where $\mathcal{X}$ is defined in~\eqref{cluster-assignment-sets},
   in time ${O({k}\cdot{m})}$.
\end{lemma}
A proof is provided in Appendix, Sec.~\ref{sec:proof-correctness-of-local-solver}.

Note that in our case,
Alg.~\ref{algo:zero-one-orthogonal} is used as both $\mathrm{P}_{\mathcal{X}}(\cdot)$ and $\mathrm{P}_{\mathcal{Y}}(\cdot)$.
Putting the pieces together, we obtain the core of our \mbox{$k$-\bcc} algorithm, outlined in Alg.~\ref{algo:bcc-core}.
Given a bipartite graph~$G$, with bi-adjacency matrix $\mathbf{B}$,
we first compute a rank-$r$ approximation~$\widetilde{\mathbf{B}}$ of $\mathbf{B}$ via the truncated SVD.
Using Alg.~\ref{algo:bilinear-lowrank},
equipped with Alg.~\ref{algo:zero-one-orthogonal} as a subroutine,
we approximately solve the bilinear maximization~\eqref{bcc:maxagree} with
argument~$\widetilde{\mathbf{B}}$.
The output is a pair of valid cluster assignment matrices for~$U$ and~$V$.

\begin{algorithm*}[t!]
   \caption{k-BCC via Low-rank Bilinear Maximization}
   \label{algo:bcc-core}
   {\color{black}    \small
   \begin{algorithmic}[1]
   \INPUT:
   {\textbullet} Bi-adjacency matrix $\mathbf{B} \in \lbrace \pm 1\rbrace^{m,n}$ of bipartite graph $G=(U,V,E)$,\\
   \quad\,\,\,\,\, {\textbullet} Target number of clusters $k$,
   \quad {\textbullet} Approximation rank $r$,
   \quad {\textbullet} Accuracy parameter $\epsilon \in (0, 1)$.
   \OUTPUT Clustering $\widetilde{\mathcal{C}}^{\kexp}$ of $U \cup V$ into $k$ clusters --\\
   \quad\,\,\,\,\, Equivalently, cluster membership matrices $\Xtilde^{\kexp} \in \mathcal{X}$ and $\Ytilde^{\kexp} \in \mathcal{Y}$ for $U$ and $V$, respectively.
   \STATE $\widetilde{\mathbf{B}} \gets \textsf{SVD}(\mathbf{B}, r)$
   \hfill\COMMENT{Truncated SVD.}
   \STATE Approx. solve the bilinear maximization~\eqref{bcc:maxagree} on $\widetilde{\mathbf{B}}$, over the sets $\mathcal{X}$, $\mathcal{Y}$ of valid cluster assignment matrices (Eq.~\ref{cluster-assignment-sets}):\\
   $\Xtilde^{\kexp}, \Ytilde^{\kexp} \gets {\lowrankalgoname}\bigl( \widetilde{\mathbf{B}}, \epsilon, r \bigr)$
   \hfill \COMMENT{Alg.~\ref{algo:bilinear-lowrank},\; 
   using Alg.~\ref{algo:zero-one-orthogonal} as $\mathrm{P}_{\mathcal{X}}(\cdot)$ and $\mathrm{P}_{\mathcal{Y}}(\cdot)$.}
   \end{algorithmic}
   }
\end{algorithm*}
Alg.~\ref{algo:bcc-core} exposes the configuration parameters $r$ and $\epsilon$
that control the performance of our bilinear solver,
and in turn the quality of the output $k$-clustering.
To simplify the description,
we also create a $k$-\bcc ``wrapper'' procedure outlined in Alg.~\ref{algo:kbcc-ptas}:
given a bound $k$ on the number of clusters and a single accuracy parameter $\delta \in (0,1)$, 
Alg.~\ref{algo:kbcc-ptas} configures and invokes Alg.~\ref{algo:bcc-core}.


\begin{theorem}
\label{kbcc:ptas-guarantees}
\emph{($k$-BCC.)}
For any instance $\bigl( G=(U, V, E), k\bigr)$ of the $k$-BCC problem with bi-adjacency matrix 
$\mathbf{B}$
and for any desired accuracy parameter $\delta \in (0, 1)$,
Algorithm~\ref{algo:kbcc-ptas} computes a clustering $\widetilde{\mathcal{C}}^{\kexp}$ of $U \cup V$ 
into at most $k$ clusters,
such that
\begin{align}
   \agree\bigl(\widetilde{\mathcal{C}}^{\kexp}\bigr)
\nonumber & \ge
   \bigl( 1- \delta \bigr) 
   \cdot
   \agree\bigl(\mathcal{C}^{\kexp}_{\star}\bigr),
\end{align}
where $\mathcal{C}^{\kexp}_{\star}$ is the optimal $k$-clustering,
in time
$
2^{O(\sfrac{k}{\delta^{2}} \cdot \log{\sfrac{\sqrt{k}}{\delta}})} \cdot (\delta^{-2}+k) \cdot {(|U|+|V|)} + \Tsvd(\delta^{-2}).
$
\end{theorem}

In the remainder of this section, we prove Theorem~\ref{kbcc:ptas-guarantees}.
We begin with the core Alg.~\ref{algo:bcc-core},
which invokes the low-rank bilinear solver Alg.~\ref{algo:bilinear-lowrank} with accuracy parameters $\epsilon$ and $r$ on the rank-$r$ matrix $\widetilde{\mathbf{B}}$ and computes a pair of valid cluster assignment matrices 
$\Xtilde^{\kexp}$, $\Ytilde^{\kexp}$.
By Lemma~\ref{low-rank-solver-guarantees},
and taking into account that $\sigma_{1}(\widetilde{\mathbf{B}}) = \sigma_{1}(\mathbf{B}) \le  \|\mathbf{B}\|_{\frob} \le \sqrt{mn}$, 
and the fact that
$\|\mathbf{X}\|_{F}=\sqrt{m}$
and
$\|\mathbf{Y}\|_{F}=\sqrt{n}$
for all valid cluster assignment matrix pairs,
the output pair satisfies
\begin{align}	 
   \resizebox{0.96\hsize}{!}{$
   \trace\bigl( \widetilde{\mathbf{X}}_{\star}^{\kexp}{}^{\transpose} \widetilde{\mathbf{B}} \widetilde{\Y}_{\star}^{\kexp}{} \bigr)
   -
   \trace\bigl( \widetilde{\mathbf{X}}^{\kexp}{}^{\transpose} \widetilde{\mathbf{B}} \widetilde{\Y}^{\kexp}{} \bigr)
   \le
      {2 \cdot \epsilon \cdot \sqrt{k} \cdot {mn}},
	  \nonumber
   $}
\end{align}
where 
$
	\bigl( \widetilde{\mathbf{X}}_{\star}^{\kexp}, \widetilde{\Y}_{\star}^{\kexp}{} \bigr)
	\eqdef
	\max_{\mathbf{X} \in \mathcal{X}, \mathbf{Y} \in \mathcal{Y}}
	\trace\bigl( {\mathbf{X}}^{\transpose} \widetilde{\mathbf{B}} {\Y} \bigr).
$

In turn,
by Lemma~\ref{full-rank-guarantees}
taking into account that
$
   \|\mathbf{B} - \widetilde{\mathbf{B}}\|_{2}
   \le \sqrt{mn}/\sqrt{r}
$
(see Cor.~\ref{sigma-bound}),
on the original bi-adjacency matrix~$\mathbf{B}$
the output pair satisfies
\begin{align}
   &
   \trace\left(\X_{\star}^{\kexp}{}^{\transpose}\mathbf{B}\Y_{\star}^{\kexp}{}\right)
   -
   \trace\bigl(\widetilde{\mathbf{X}}^{\kexp}{}^{\transpose}\mathbf{B}\widetilde{\Y}^{\kexp}{}\bigr)
   \nonumber\\&\qquad \qquad \le
   2  \epsilon  \sqrt{k} \cdot {mn}
   +
   2  ({r+1})^{-{1}/{2}} \cdot {mn},
   \label{bcc:proof:original-parametrized}
\end{align}
where
$
	\bigl( {\mathbf{X}}_{\star}^{\kexp}, {\Y}_{\star}^{\kexp}{} \bigr)
	\eqdef
	\max_{\mathbf{X} \in \mathcal{X}, \mathbf{Y} \in \mathcal{Y}}
	\trace\bigl( {\mathbf{X}}^{\transpose} {\mathbf{B}} {\Y} \bigr).
$
The unknown optimal pair ${\mathbf{X}}_{\star}^{\kexp}, {\Y}_{\star}^{\kexp}{}$ represents the optimal $k$-clustering $\mathcal{C}_{\star}^{\kexp}$,
\textit{i.e.}, the clustering that achieves the maximum number of agreements using at most $k$ clusters.
Similarly, let $\widetilde{\mathcal{C}}^{\kexp}$ be the clustering induced by the computed pair
$\widetilde{\mathbf{X}}^{\kexp}{}$, $\widetilde{\Y}^{\kexp}{}$.
From~\eqref{bcc:proof:original-parametrized} and Lemma~\ref{lemma:num-agreements}, it immediately follows that
\begin{align}
   \agree\bigl(\mathcal{C}_{\star}^{\kexp}\bigr)
   &-
   \agree\bigl(\widetilde{\mathcal{C}}^{\kexp}\bigr)
   \nonumber\\ & \quad
   \le
    2 \cdot ( \epsilon  \sqrt{k} + \sqrt{{r+1}} ) \cdot {mn}.
   \label{core-bcc:numagree-guarantee}
\end{align}
Eq.~\ref{core-bcc:numagree-guarantee} establishes the guarantee of Alg.~\ref{algo:bcc-core} as a function of its accuracy parameters $r$ and $\epsilon$.
Substituting those parameters by the values assigned by Alg.~\ref{algo:kbcc-ptas},
we obtain
\begin{align}
   \agree\bigl(\mathcal{C}_{\star}^{\kexp}\bigr)
   -
   \agree\bigl(\widetilde{\mathcal{C}}^{\kexp}\bigr)
   &\le
     2^{-1} \cdot \delta \cdot {mn}.
     \label{kbcc:bound-with-delta}
\end{align}

\begin{fact}
\label{bcc:agree:lb}
For any \bcc instance {(or $k$-\bcc with ${k\ge 2}$)},
the optimal clustering achieves at least ${m}{n}/2$ agreements.
\end{fact}
\vspace{-10pt}
\begin{proof}
If more than half of the edges are labeled `\plus', a single cluster containing all vertices achieves at least $nm/2$ agreements. 
Otherwise, two clusters corresponding to $U$ and $V$ achieve the same result.
\end{proof}

In other words,
Fact~\ref{bcc:agree:lb} states that
\begin{align}
   \agree\bigl(\mathcal{C}^{\kexp}_{\star}\bigr)
   \ge
   mn/2, \quad \forall k \ge 2.
   \label{fact-made-simple}
\end{align}
Combining~\eqref{fact-made-simple} with~\eqref{kbcc:bound-with-delta},
we obtain
\begin{align}
   \agree\bigl(\widetilde{\mathcal{C}}^{\kexp}\bigr)
   &\ge
   (1-\delta) \cdot \agree\bigl(\mathcal{C}_{\star}^{\kexp}\bigr),
   \nonumber
\end{align}
which is the desired result.
Finally,
the computational complexity is obtained substituting the appropriate values in that of Alg.~\ref{algo:bilinear-lowrank} as in Lemma~\ref{low-rank-solver-guarantees}, and the time complexity of subroutine Alg.~\ref{algo:zero-one-orthogonal} given in Lemma~\ref{lemma:local-solver-for-bcc}.
This completes the proof of Theorem~\ref{kbcc:ptas-guarantees}.

\section{An Efficient PTAS for BCC}
We provide an efficient polynomial time approximation scheme (EPTAS) for {\bcc}/{\maxagree},
\textit{i.e.}, the unconstrained version with no upper bound on the number of output clusters.
We rely on the following key observation:
\emph{any constant factor approximation of the \maxagree objective, can be achieved by constant number of clusters}.%
\footnote{%
A similar result for {\cc}/{\maxagree} exists in~\cite{bansal2004correlation}.
} 
Hence, the desired approximation algorithm for \bcc can be obtained by invoking the $k$-\bcc algorithm under the appropriate configuration.

Recall that in the unconstrained \bcc setting,
the number of output clusters is an optimization variable;
the optimal clustering may comprise any number of clusters, from a single cluster containing all vertices of the graph,
up to $|U|+|V|$ singleton clusters.
In principle, one could solve $\maxagree[k]$ for all possible values of the parameter~$k$ to identify the best clustering,
but this approach is computationally intractable.

On the contrary, if one is willing to settle for an approximately optimal solution,
then a constant number of clusters may suffice.
More formally,
\begin{lemma}
   \label{lemma:bcc-constant-num-clusters}
   For any \bcc instance, and ${0 < \epsilon \le 1}$,
   there exists a clustering $\mathcal{C}$ with at most $k = 2\cdot\epsilon^{-1}+2 $ clusters such that
   $\agree(\mathcal{C}) \ge  \agree(\mathcal{C}_{\star})- \epsilon \cdot {n}{m}$,
   where
   $\mathcal{C}_{\star}$ denotes the optimal clustering.
\end{lemma}
We defer the proof to the end of the section.

In conjunction with Fact~\ref{bcc:agree:lb},
Lemma~\ref{lemma:bcc-constant-num-clusters}  suggests that to obtain a constant factor approximation for the unconstrained \bcc/\maxagree problem, for any constant arbitrarily close to $1$, 
it suffices to solve a $k$-\bcc instance for a sufficiently large --but constant-- number of clusters~$k$.
In particular, to obtain an ${(1-\delta)}$-factor approximation
for any constant ${\delta \in (0,1)}$,
it suffices to solve the $k$-\bcc problem with an upper bound ${k=O(\delta^{-1})}$ on the number of clusters.

\begin{algorithm}[tb!]
   \caption{$k$-BCC/\maxagree}
   \label{algo:kbcc-ptas}
   {\color{black}    \small
   \begin{algorithmic}[1]
   \INPUT:
   Bi-adjacency matrix $\mathbf{B} \in \lbrace \pm 1\rbrace^{m,n}$,\\
   \quad\,\,\,\,\, Target number of  $k$,\\
   \quad\,\,\,\,\, Accuracy $\delta \in (0,1)$.
   \OUTPUT Clustering $\widetilde{\mathcal{C}}^{\kexp}$ of $U \cup V$ such that\\
   \quad\,\,\,\,\, $
      \agree\bigl(\widetilde{\mathcal{C}}^{\kexp}\bigr)
      \ge
      \bigl( 1- {\delta} \bigr) 
      \cdot
      \agree\bigl(\mathcal{C}^{\kexp}_{\star}\bigr).
      $
   \STATE Set up parameters:\\
   $\epsilon \gets 2^{-3} \cdot \delta \cdot k^{-1/2}$, \,
   $r \gets 2^{6} \cdot \delta^{-2} - 1$.
   \STATE Return output of Alg.~\ref{algo:bcc-core}
   for input $(\mathbf{B}, k, r, \epsilon)$.
   \end{algorithmic}
   }
\end{algorithm}

Alg.~\ref{algo:bcc-ptas} outlines the approximation scheme for \bcc.
For a given accuracy parameter~$\delta$, it invokes 
Alg.~\ref{algo:bcc-core} under an appropriate configuration of the parameters $k$, $r$ and $\epsilon$,
yielding the following guarantees:
\begin{theorem}
   \label{bcc:algo-guarantees}
   \emph{(PTAS for BCC.)}
   For any instance $G=(U, V, E)$ of the \bcc problem with bi-adjacency matrix~$\mathbf{B}$
   and for any desired accuracy parameter $\delta \in (0, 1)$,
   Algorithm~\ref{algo:bcc-ptas} computes a clustering $\widetilde{\mathcal{C}}$ of $U \cup V$ 
   into (at most) ${2^{3} \cdot \delta^{-1}}$ clusters,
   such that
   \begin{align}
      \agree\bigl(\widetilde{\mathcal{C}}\bigr)
   \nonumber & \ge
      \bigl( 1- {\delta} \bigr) 
      \cdot
      \agree\bigl(\mathcal{C}_{\star}\bigr),
   \end{align}
   where $\mathcal{C}_{\star}$ is an optimal clustering (with no constraint on the number of clusters),
   in time
   $$
   2^{O(\delta^{-3} \cdot \log{\delta^{-3}})} \cdot \delta^{-2} \cdot {(m+n)} + \Tsvd(\delta^{-2}).
   $$
\end{theorem}

\begin{algorithm}[t!]
   \caption{A PTAS for \bcc/\maxagree}
   \label{algo:bcc-ptas}
   {\color{black}    \small
   \begin{algorithmic}[1]
   \INPUT:
   Bi-adjacency matrix $\mathbf{B} \in \lbrace \pm 1\rbrace^{m,n}$,\\
   \quad\,\,\,\,\, Accuracy $\delta \in (0,1)$.
   \OUTPUT Clustering $\widetilde{\mathcal{C}}$ of $U \cup V$ (into at most\\
   \quad\,\,\,\,\, $2^{3} \cdot \delta^{-1}$ clusters) such that\\
   \quad\,\,\,\,\, 
   $
   \agree\bigl(\widetilde{\mathcal{C}}\bigr)
   \ge
   \bigl( 1- {\delta} \bigr) 
   \cdot
   \agree\bigl(\mathcal{C}_{\star}\bigr).
   $
   \STATE Set up parameters:\\
   $k \gets 2^{3} \cdot \delta^{-1}$, \,
   $\epsilon \gets 2^{-6} \cdot \delta^{2}$, \,
   $r \gets 2^{8} \cdot \delta^{-2} - 1$.
   \STATE Return output of Alg.~\ref{algo:bcc-core} for input $(\mathbf{B}, k, r, \epsilon)$.
   \end{algorithmic}
   }
\end{algorithm}
The proof of Theorem~\ref{bcc:algo-guarantees} follows from the guarantees of Alg.~\ref{algo:bcc-core} in~\eqref{core-bcc:numagree-guarantee} substituting the values of the parameters $k$, $r$ and $\epsilon$  with the values specified by Alg~\ref{algo:bcc-ptas}.
The core $k$-\bcc Alg.~\ref{algo:bcc-core} returns a clustering~$\widetilde{\mathcal{C}}$ with at most $k_0=2^{3} \cdot \delta^{-1}$ clusters that satisfies
\begin{align}
   \agree\bigl(\widetilde{\mathcal{C}}\bigr)
   &\ge
   \agree\bigl(\mathcal{C}_{\star}^{\mathsmaller{(k_0)}}\bigr)
   - \tfrac{\delta}{4} \cdot {mn},
\label{bcc:trace-guarantee-final-fixed-k}
\end{align}
where $\mathcal{C}_{\star}^{\mathsmaller{(k_0)}}$ is the best among the clusterings using at most $k_0$ clusters.
Also, 
for $k=k_0$,
Lemma~\ref{lemma:bcc-constant-num-clusters} implies that
$\agree(\mathcal{C}_{\star}^{\mathsmaller{(k_0)}}) \ge  \agree(\mathcal{C}_{\star})- \sfrac{\delta}{4} \cdot {n}{m}$, where~$\mathcal{C}_{\star}$ is the optimal clustering without any constraint on the number of output clusters.
Hence, 
\begin{align}
   \agree\bigl(\widetilde{\mathcal{C}}\bigr)
& \ge
   \agree\bigl(\mathcal{C}_{\star}\bigr)
   -
   \sfrac{\delta}{2} \cdot {m}{n}.
\label{bcc:clustering-compared-to-optimal}
\end{align}
Continuing from~\eqref{bcc:clustering-compared-to-optimal},
and taking into account Fact~\ref{bcc:agree:lb},
we conclude that the output $\widetilde{\mathcal{C}}$ of Alg.~\ref{algo:bcc-ptas} satisfies
\begin{align}
   \agree\bigl(\widetilde{\mathcal{C}}\bigr)
 \ge
   ( 1- \delta) 
   \cdot
   \agree\bigl(\mathcal{C}_{\star}\bigr),
\end{align}
which is the desired result.
Finally, the computational complexity of Alg.~\ref{algo:bcc-ptas} follows from that of Alg.~\ref{algo:bcc-core} substituting the parameter values in Lemma~\ref{low-rank-solver-guarantees}, taking into account the running time $\Tx = O(\delta^{-2} \cdot {m})$
of Alg.~\ref{algo:zero-one-orthogonal} used as subroutine $\mathrm{P}_{\mathcal{X}}(\cdot)$ 
and similarly $\Ty = O(\delta^{-2} \cdot {n})$ for 
$\mathrm{P}_{\mathcal{Y}}(\cdot)$.
This concludes the proof of Theorem~\ref{bcc:algo-guarantees}.

For any desired constant accuracy ${\delta \in (0,1)}$,
Alg.~\ref{algo:bcc-ptas} outputs a clustering that achieves a number of agreements within a $(1-\delta)$-factor from the optimal,
in time that grows exponentially in $\delta^{-1}$, but linearly in the size~${mn}$ of the input.
In other words, Alg.~\ref{algo:bcc-ptas} is an EPTAS for \bcc/\maxagree.

We complete the section with a proof for Lemma~\ref{lemma:bcc-constant-num-clusters},
which stated that a constant number of clusters suffices to obtain an approximately optimal solution.

\subsection{Proof of Lemma~\ref{lemma:bcc-constant-num-clusters}}
\label{sec:proof-constant-num-of-clusters}
It suffices to consider $\epsilon > 2/{({m+n}-2)}$
as the lemma holds trivially otherwise. 
Without loss of generality, we focus on clusterings whose
all but at most two clusters contain vertices from both $U$ and $V$;
one of the two remaining clusters can contain vertices only from $U$ and the other only from $V$.
To verify that, consider an arbitrary clustering $\mathcal{C}$
and let $\mathcal{C}^{\prime}$ be the clustering obtained by merging 
all clusters of $\mathcal{C}$ containing only vertices of $U$ into a single cluster, and those containing only vertices of $V$ into another.
The number of agreements is not affected, \textit{i.e.},
 $\agree(\mathcal{C}^{\prime}) = \agree(\mathcal{C})$.

We show that a clustering $\mathcal{C}$ with the properties described in the lemma exists
by modifying~$\mathcal{C}_{\star}$.
Let
   $\mathcal{S}_{U} \subseteq \mathcal{C}_{\star}$ be the set of clusters that contain at most $\epsilon{m}/2$ vertices of $U$, and 
   $\mathcal{S}_{V} \subseteq \mathcal{C}_{\star}$ those containing at most $\epsilon{n}/2$ vertices of $V$.
($\mathcal{S}_{U}$ and $\mathcal{S}_{V}$ may not be disjoint).
Finally, let $\mathcal{B}$ be the set of vertices contained in clusters of $\mathcal{S}_{U} \cup \mathcal{S}_{V}$.
We construct $\mathcal{C}$ as follows.
Clusters of~$\mathcal{C}_{\star}$ not in $\mathcal{S}_{U} \cup \mathcal{S}_{V}$
are left intact. 
Each such cluster has size at least $\epsilon {(n+m)}/2$.
Further, all vertices in~$\mathcal{B}$ are rearranged into two clusters:
one for the vertices in $\mathcal{B} \cap U$
and one those in $\mathcal{B} \cap V$.

The above rearrangement can reduce the number of agreements by at most $\epsilon{n}{m}$.
To verify that, consider a cluster $C$ in $\mathcal{S}_{U}$;
the cluster contains at most $\epsilon{m}/2$ vertices of $U$.
Let $t \eqdef \lvert V \cap C \rvert$.
Splitting~$C$ into two smaller clusters $C \cap U$ and $C \cap V$
can reduce the number of agreements by $\frac{1}{2}\epsilon {m} t$, \textit{i.e.}, the total number of edges in~$C$. 
Note that agreements on $-$ edges are not affected by splitting a cluster.
Performing the same operation on all clusters in~$\mathcal{S}_{U}$ incurs a total reduction
\begin{align}
 \tfrac{1}{2} \epsilon {m} \sum_{C \in \mathcal{S}_{U}} \lvert V \cap C \rvert
& \le
 \tfrac{1}{2} \epsilon{m} |V|
 =
 \tfrac{1}{2} \epsilon{m}{n}.
 \nonumber
\end{align}
Repeating on $\mathcal{S}_{V}$ incurs an additional cost of at most $\epsilon{m}{n}/2$ agreements,
while merging all clusters containing only vertices from $U$ (similarly for $V$) into a single cluster does not reduce the agreements. 
We conclude that
 $
   \agree(\mathcal{C})
   \ge
   \agree(\mathcal{C}_{\star}) - \epsilon{nm}.
 $

Finally, by construction all but at most two clusters in~$\mathcal{C}$ contain at least $\frac{1}{2}\epsilon {(n+m)}$.
In turn,
 $$
 n + m
 = \sum_{C \in \mathcal{C}} |C|
 \ge 
 \left( |\mathcal{C}|-2 \right) \cdot \tfrac{1}{2} \epsilon {(n+m)},
 $$
 from which the desired result follows.
 \qed

\section{Experiments}

We evaluate our algorithm on synthetic and real data
and compare with PivotBiCluster~\cite{ailon2012improved}
and the SDP-based approach of~\cite{swamy2004correlation}%
\footnote{
  All algorithms are prototypically implemented in Matlab. 
  ~\cite{swamy2004correlation} was implemented using CVX. 
}
Although \cite{amit2004bicluster,charikar2003clustering, demaine2006correlation} provide algorithms applicable to \bcc,
those rely on LP or SDP formulations with a high number of constraints
which imposes limitations in practice.

We run our core $k$-\bcc algorithm (Alg.~\ref{algo:bcc-core}) to obtain a $k$-clustering.
Disregarding the theoretical guarantees, we apply a threshold on the execution time which can only hurt the performance; equivalently, the iterative procedure of Alg.~\ref{algo:bilinear-lowrank} is executed for an arbitrary subset of the points in the $\epsilon$-net (Alg.~\ref{algo:bilinear-lowrank}, step~$3$).

\subsection{Synthetic Data}
We generate synthetic \bcc instances as follows.
We arbitrarily set $m=100$, $n=50$, and number of clusters~$k^{\prime}=5$,
and construct a complete bipartite graph ${G=(U,V,E)}$ with $|U|=m$, $|V|=n$ and assign binary $\pm 1$ labels to the edges according to a random $k^{\prime}$-clustering of the vertices.
Then, we modify the sign of each label independently with probability~$p$.
We consider multiple values of $p$ in the range $[0, 0.5]$.
For each value, we generate $10$ random instances as described above and compute a vertex clustering using all algorithms.

\setlength{\belowcaptionskip}{-10pt}
\begin{figure}[t!]
\centering
\includegraphics[width=\columnwidth]{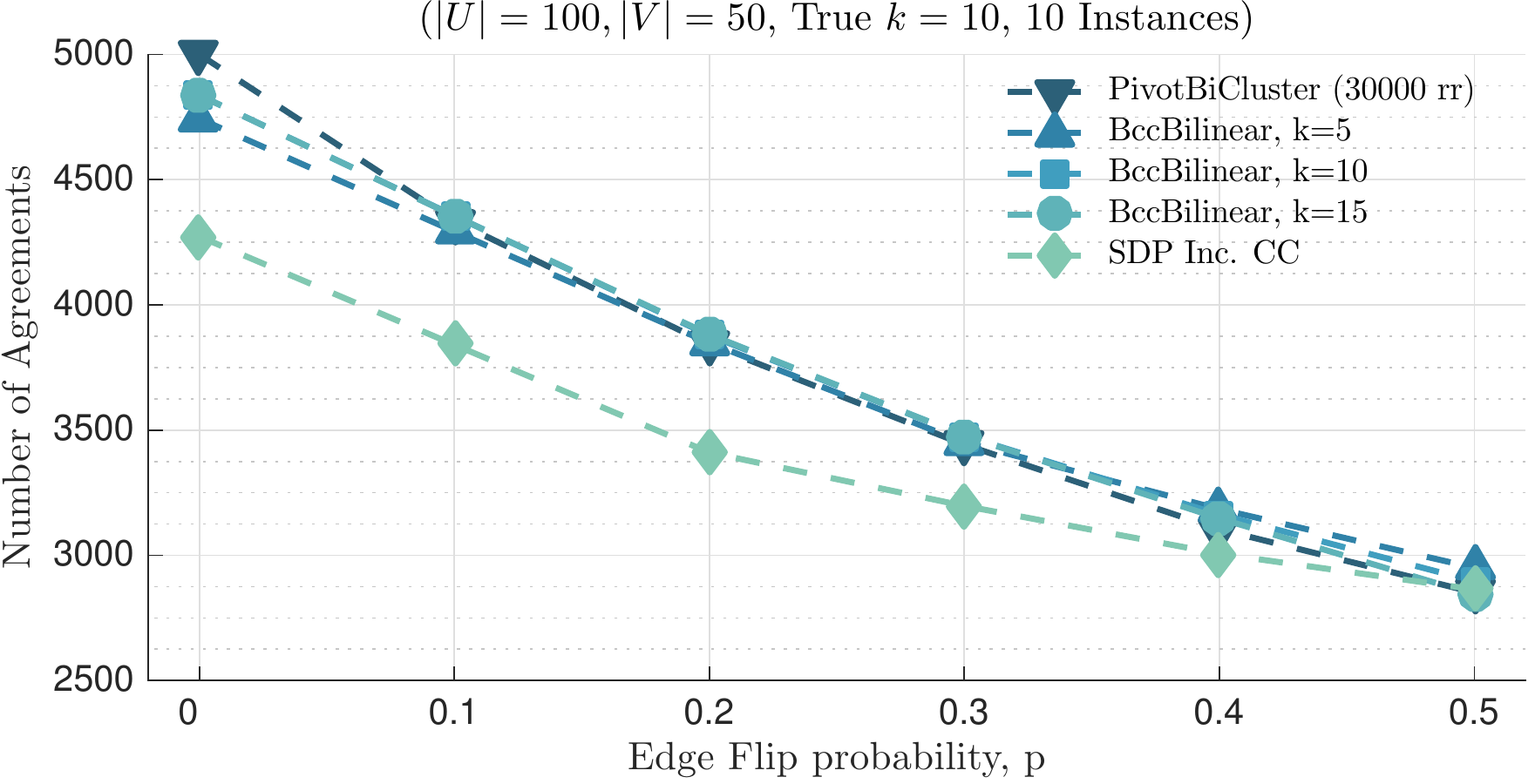}
\includegraphics[width=\columnwidth]{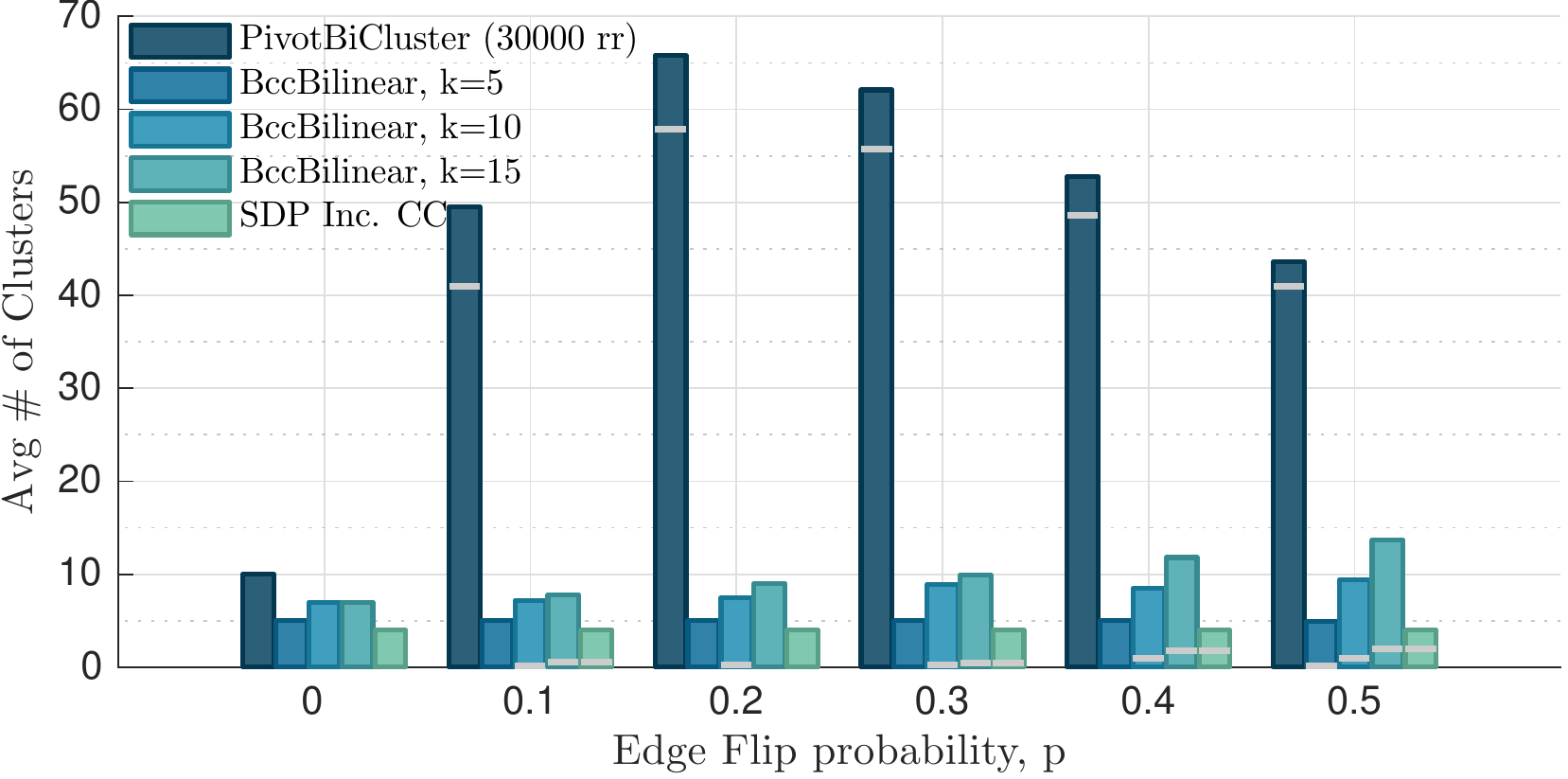}

\begin{center}
{
\footnotesize
  \rowcolors{2}{black!05!white}{white}
  \begin{tabular}{lr@{}rrrrrr}
  & $p=$ & $0.0$ & $0.1$ & $0.2$ & $0.3$ & $0.4$ & $0.5$ \\
  \toprule
  \multicolumn{2}{l}{PivotBiCluster~\cite{ailon2012improved}}
          & 124  &  83  &  66  &  54  &  45  &  39\\
  \multicolumn{2}{l}{BccBil. (k=5)} 
          & 69  &  70  &  70  &  71  &  71  &  71\\
  \multicolumn{2}{l}{BccBil. (k=10)}
          & 73  &  74  &  75  &  75  &  75  &  75\\
  \multicolumn{2}{l}{BccBil. (k=15)}
          & 78  &  79  &  79  &  80  &  80  &  80\\
  \multicolumn{2}{l}{SDP Inc CC~\cite{swamy2004correlation}}
          & 198 & 329  & 257 &  258 &  266  & 314\\
  \end{tabular}
  Average runtimes/instance (in seconds).
}
\end{center}
\caption{
	Synthetic data.
	We generate an $m\times n$ bipartite graph with edge weights $\pm{1}$ according to a random vertex $k$-clustering
	and subsequently flip the sign of each edge with probability~$p$.
	We plot the average number of agreements achieved by each method  over $10$ random instances for each $p$ value.
	The bar plot depicts the average number of output clusters for each scheme/$p$-value pair. 
  Within each bar, a horizontal line marks the number of singleton clusters.
}
\label{fig:synthetic}
\end{figure}

PivotBiCluster returns a clustering of arbitrary size; no parameter $k$ is specified and the algorithm determines the number of clusters in an attempt to maximize the number of agreements.
The majority of the output clusters are singletons which can be merged into two clusters (see Subsec.~\ref{sec:proof-constant-num-of-clusters}).
The algorithm is substantially faster that the other approaches on examples of this scale. Hence, we run PivotBiCluster with $3 \cdot 10^4$ random restarts on each instance and depict best results to allow for comparable running times.
Contrary to PivotBiCluster, for our algorithm, which is reffered to as BccBilinear, we need to specify the target number $k$ of clusters.
We run it for $k=5, 10$ and $15$ and arbitrarily set the parameter $r=5$ and terminate our algorithm after $10^{4}$ random samples/rounds.
Finally, the SDP-based approach returns $4$ clusters by construction, and is configured to output best results among $100$ random pivoting steps.

Fig.~\ref{fig:synthetic} depicts the number of agreements achieved by each algorithm for each value of~$p$, the number of output clusters, and the execution times.
All numbers are averages over multiple random instances.

We note that in some cases
and contrary to intuition, our algorithm performs better for lower values of $k$, the target number of clusters.
We attribute this phenomenon to the fact that for higher values of $k$, we should typically use higher approximation rank and consider larger number of samples.

\subsection{MovieLens Dataset}

The MovieLens datasets~\cite{movielensurl} are sets of movie ratings:
each of $m$ users assigns scores in $\lbrace {1, \dots, 5} \rbrace$ to a small subset of the $n$ movies.
Table~\ref{table:movielens:info} lists the dimensions of the datasets.
From each dataset, we generate an instance of the (incomplete) \bcc problem as follows:
we construct the bipartite graph on the user-movie pairs using the ratings as edge weights, compute the average weight and finally set weights higher than the average to $+1$ and the others to $-1$.
\begin{table}[t!]
\centering
\footnotesize
\rowcolors{2}{white}{black!05!white}
\begin{tabular}{lccc}
  \toprule
  Dataset & $m$ (Users) & $n$ (Movies) & Ratings \\ 
  \cmidrule{1-1} \cmidrule{2-2} \cmidrule{3-3} \cmidrule{4-4}
   MovieLens$100\mathrm{K}$ & $1000$ & $1700$ & $10^5$ \\
   MovieLens$1\mathrm{M}$ & $6000$ & $4000$ & $10^6$ \\
   MovieLens$10\mathrm{M}$ & $72000$ & $10000$ & $10^7$ \\
  \bottomrule
\end{tabular}
\caption{Summary of MovieLens datasets~\cite{movielensurl}.} 
\label{table:movielens:info}
\end{table}

We run our algorithm to obtain a clustering of the vertices.
For a reference, we compare to PivotBiCluster with $50$ random restarts.
Note, however, that PivotBiCluster is not designed for the incomplete \bcc problem; to apply the algorithm, we effectively treat missing edges as edges of negative weight.
Finally, the SDP approach of~\cite{swamy2004correlation}, albeit suitable for the incomplete \cc problem, does not scale to this size of input.
Table~\ref{table:movielens:results} lists the number of agreements achieved by each method on each one of the three datasets and the corresponding execution times.

\begin{table}[b!]
\centering
\footnotesize
\begin{tabular}{lccc}
  \toprule
  MovieLens & $100\mathrm{K}$ & $1\mathrm{M}$ & $10\mathrm{M}$ \\ 
  \cmidrule{1-1} \cmidrule{2-2} \cmidrule{3-3} \cmidrule{4-4}
   PivotBiCluster & 
    $46134$&
    $429277$&
    $5008577$\\
    &
    \color{darkgray}$(27.95)$ &
    \color{darkgray}$(651.13)$ &
    \color{darkgray}$(1.5\cdot 10^5)$\\
   BccBilinear &
    $68141$ &
    $694366$ & 
    $6857509$ \\
    &
    \color{darkgray}$(6.65)$ &
    \color{darkgray}$(19.50)$&
    \color{darkgray}$(1.2\cdot 10^3)$\\
  \bottomrule
\end{tabular}
\caption{Number of agreements achieved by the two algorithms on incomplete $(k)$-\bcc instances obtained from the MovieLens datasets~\cite{movielensurl}.
For PivotBiCluster we present best results over $50$ random restarts.
Our algorithm was arbitrarily configured with ${r=4}$, ${k=10}$ and a limit of ${10^4}$ samples (iterations).
We also note in parenteses the runtimes in seconds.
} 
\label{table:movielens:results}
\end{table}

\section{Conclusions}

We presented the first algorithm with provable approximation guarantees for $k$-\bcc/\maxagree.
Our approach relied on formulating $k$-\bcc as a constrained bilinear maximization over the sets of cluster assignment matrices and developing a simple framework to approximately solve that combinatorial optimization.

In the unconstrained \bcc setting, with no bound on the number of output clusters, 
we showed that any constant multiplicative factor approximation for the \maxagree objective can be achieved using a constant number of clusters.
In turn, under the appropriate configuration, our $k$-\bcc algorithm
yields an Efficient PTAS for \bcc/\maxagree.

\textbf{Acknowledgments}\;
This research has been supported by NSF Grants CCF $1344179$, $1344364$, $1407278$,
$1422549$ and ARO YIP W911NF-14-1-0258.
DP is supported by NSF awards CCF-1217058 and CCF-1116404 and MURI AFOSR
grant 556016. 

{%
\small
\fontsize{10}{11.8}\selectfont
\balance
\bibliographystyle{plain}
\bibliography{bcc}

\begin{thebibliography}{10}

\bibitem{ahn2015correlation}
KookJin Ahn, Graham Cormode, Sudipto Guha, Andrew McGregor, and Anthony Wirth.
\newblock Correlation clustering in data streams.
\newblock In {\em Proceedings of the 32nd International Conference on Machine
  Learning (ICML-15)}, pages 2237--2246, 2015.

\bibitem{ailon2012improved}
Nir Ailon, Noa Avigdor-Elgrabli, Edo Liberty, and Anke Van~Zuylen.
\newblock Improved approximation algorithms for bipartite correlation
  clustering.
\newblock {\em SIAM Journal on Computing}, 41(5):1110--1121, 2012.

\bibitem{ailon2008aggregating}
Nir Ailon, Moses Charikar, and Alantha Newman.
\newblock Aggregating inconsistent information: ranking and clustering.
\newblock {\em Journal of the ACM (JACM)}, 55(5):23, 2008.

\bibitem{ailon2009correlation}
Nir Ailon and Edo Liberty.
\newblock Correlation clustering revisited: The ``true'' cost of error
  minimization problems.
\newblock In {\em Automata, Languages and Programming}, pages 24--36. Springer,
  2009.

\bibitem{amit2004bicluster}
Noga Amit.
\newblock The bicluster graph editing problem.
\newblock Master's thesis, Tel Aviv University, 2004.

\bibitem{asteris2015sparseviabipartite}
Megasthenis Asteris, Dimitris Papailiopoulos, Anastasios Kyrillidis, and
  Alexandros~G Dimakis.
\newblock Sparse pca via bipartite matchings.
\newblock In {\em Advances in Neural Information Processing Systems 28}, pages
  766--774. Curran Associates, Inc., 2015.

\bibitem{bansal2004correlation}
Nikhil Bansal, Avrim Blum, and Shuchi Chawla.
\newblock Correlation clustering.
\newblock {\em Machine Learning}, 56(1-3):89--113, 2004.

\bibitem{bonchi2014correlation}
Francesco Bonchi, David Garcia-Soriano, and Edo Liberty.
\newblock Correlation clustering: from theory to practice.
\newblock In {\em Proceedings of the 20th ACM International Conference on
  Knowledge Discovery and Data mining}, pages 1972--1972. ACM, 2014.

\bibitem{charikar2003clustering}
Moses Charikar, Venkatesan Guruswami, and Anthony Wirth.
\newblock Clustering with qualitative information.
\newblock In {\em Foundations of Computer Science, 2003. Proceedings. 44th
  Annual IEEE Symposium on}, pages 524--533. IEEE, 2003.

\bibitem{cohen2002learning}
William~W Cohen and Jacob Richman.
\newblock Learning to match and cluster large high-dimensional data sets for
  data integration.
\newblock In {\em Proceedings of the 8th ACM International Conference on
  Knowledge Discovery and Data mining}, pages 475--480. ACM, 2002.

\bibitem{demaine2006correlation}
Erik~D Demaine, Dotan Emanuel, Amos Fiat, and Nicole Immorlica.
\newblock Correlation clustering in general weighted graphs.
\newblock {\em Theoretical Computer Science}, 361(2):172--187, 2006.

\bibitem{fern2004solving}
Xiaoli~Zhang Fern and Carla~E Brodley.
\newblock Solving cluster ensemble problems by bipartite graph partitioning.
\newblock In {\em Proceedings of the 21st International Conference on Machine
  Learning}, page~36. ACM, 2004.

\bibitem{giotis2006correlation}
Ioannis Giotis and Venkatesan Guruswami.
\newblock Correlation clustering with a fixed number of clusters.
\newblock In {\em Proceedings of the 17th annual ACM-SIAM Symposium on Discrete
  algorithms}, pages 1167--1176. Society for Industrial and Applied
  Mathematics, 2006.

\bibitem{movielensurl}
University of~Minnesota GroupLens~Lab.
\newblock Movielens datasets.
\newblock \url{http://grouplens.org/datasets/movielens/}, 2015.
\newblock Accessed: 2015-10-03.

\bibitem{karpinski2009linear}
Marek Karpinski and Warren Schudy.
\newblock Linear time approximation schemes for the gale-berlekamp game and
  related minimization problems.
\newblock In {\em Proceedings of the 41st annual ACM Symposium on Theory of
  Computing}, pages 313--322. ACM, 2009.

\bibitem{madeira2004biclustering}
Sara~C Madeira and Arlindo~L Oliveira.
\newblock Biclustering algorithms for biological data analysis: a survey.
\newblock {\em IEEE/ACM Transactions on Computational Biology and
  Bioinformatics}, 1(1):24--45, 2004.

\bibitem{swamy2004correlation}
Chaitanya Swamy.
\newblock Correlation clustering: maximizing agreements via semidefinite
  programming.
\newblock In {\em Proceedings of the 15th annual ACM-SIAM Symposium on Discrete
  Algorithms}, pages 526--527. Society for Industrial and Applied Mathematics,
  2004.

\bibitem{symeonidis2007nearest}
Panagiotis Symeonidis, Alexandros Nanopoulos, Apostolos Papadopoulos, and
  Yannis Manolopoulos.
\newblock Nearest-biclusters collaborative filtering with constant values.
\newblock In {\em Advances in web mining and web usage analysis}, pages 36--55.
  Springer, 2007.

\bibitem{van2009deterministic}
Anke Van~Zuylen and David~P Williamson.
\newblock Deterministic pivoting algorithms for constrained ranking and
  clustering problems.
\newblock {\em Mathematics of Operations Research}, 34(3):594--620, 2009.

\bibitem{vlachos2014improving}
Michail Vlachos, Francesco Fusco, Charalambos Mavroforakis, Anastasios
  Kyrillidis, and Vassilios~G Vassiliadis.
\newblock Improving co-cluster quality with application to product
  recommendations.
\newblock In {\em Proceedings of the 23rd ACM International Conference on
  Information and Knowledge Management}, pages 679--688. ACM, 2014.

\bibitem{zha2001bipartite}
Hongyuan Zha, Xiaofeng He, Chris Ding, Horst Simon, and Ming Gu.
\newblock Bipartite graph partitioning and data clustering.
\newblock In {\em Proceedings of the 10th ACM International Conference on
  Information and Knowledge Management}, pages 25--32. ACM, 2001.

\end{thebibliography}
}

\newpage
\clearpage
\newpage
\appendix
\section{Bilinear Maximization Guarantees}
\label{sec:approx-proof}
\begin{replemma}{low-rank-solver-guarantees}
      For any  real ${m \times n}$, rank-$r$ matrix~$\widetilde{\mathbf{A}}$ 
   and arbitrary norm-bounded sets
   ${\mathcal{X} \subset \mathbb{R}^{m \times k}}$ and
   ${\mathcal{Y} \subset \mathbb{R}^{n \times k}}$,
   let
   $$
   \bigl(\widetilde{\X}_{\star}, \widetilde{\Y}_{\star}\bigr)
   \eqdef
   \argmax_{\mathbf{X} \in \mathcal{X}, \mathbf{Y} \in \mathcal{Y}}\trace\bigl(\X^{\transpose}{\widetilde{\mathbf{A}}}\mathbf{Y}\bigr).
   $$
   If there exist operators $\mathrm{P}_{\mathcal{X}}: \mathbb{R}^{m \times k} \rightarrow \mathcal{X}$ such that
   \begin{align}
   \mathrm{P}_{\mathcal{X}}(\mathbf{L}) = \argmax_{\mathbf{X} \in \mathcal{X}} \trace\bigl(\X^{\transpose}\mathbf{L}\bigr)
   \nonumber
   \end{align}
   and similarly, $\mathrm{P}_{\mathcal{Y}}: \mathbb{R}^{n \times k} \rightarrow \mathcal{Y}$ such that
   \begin{align}
   \mathrm{P}_{\mathcal{Y}}(\mathbf{R}) = \argmax_{\mathbf{Y} \in \mathcal{Y}} \trace\bigl(\mathbf{R}^{\transpose}\mathbf{Y}\bigr)
   \nonumber
   \end{align}
   with running times $\Tx$ and $\Ty$, respectively,
   then 
   Algorithm~\ref{algo:bilinear-lowrank} outputs $\widetilde{\X} \in \mathcal{X}$ 
   and $\widetilde{\Y} \in \mathcal{Y}$ such that
   \begin{align}
     \trace\bigl( \widetilde{\X}^{\transpose} \widetilde{\mathbf{A}} \widetilde{\Y} \bigr)
     \ge
     \trace\bigl( \widetilde{\X}_{\star}^{\transpose} \widetilde{\mathbf{A}} \widetilde{\Y}_{\star} \bigr)
     - 2  \epsilon  \sqrt{k} \cdot \|\widetilde{\mathbf{A}}\|_{2} \cdot \mu_{\mathcal{X}} \cdot \mu_{\mathcal{Y}},
     \nonumber
   \end{align}
   where
   $\mu_{\mathcal{X}}\eqdef \max_{\mathbf{X} \in  \mathcal{X}}\| \mathbf{X} \|_{\frob}$
   and
   $\mu_{\mathcal{Y}}\eqdef \max_{\mathbf{Y} \in  \mathcal{Y}}\| \mathbf{Y} \|_{\frob}$,
   in time $O\mathopen{}\bigl(\bigl({2\sqrt{r}}/{\epsilon}\bigr)^{r \cdot k} \cdot \bigl(\Tx+\Ty+(m+n)r\bigr)\bigr) + \Tsvd(r)$.
\end{replemma}
\begin{proof}
   In the sequel, $\widetilde{\mathbf{U}}$, $\widetilde{\mathbf{\Sigma}}$ and $\widetilde{\mathbf{V}}$ are used to denote the $r$-truncated singular value decomposition of $\widetilde{\mathbf{A}}$.

   Without loss of generality, we assume that $\mu_{\mathcal{X}}=\mu_{\mathcal{Y}} =1$
   since 
   the variables in $\mathcal{X}$ and $\mathcal{Y}$ can be normalized by $\mu_{\mathcal{X}}$ and $\mu_{\mathcal{Y}}$, respectively, while simultaneously scaling the singular values of $\widetilde{\mathbf{A}}$ by a factor of $\mu_{\mathcal{X}} \cdot \mu_{\mathcal{Y}}$.
   Then, $\|\mathbf{Y}\|_{\infty, 2} \le 1$, $\forall\, \mathbf{Y}\in \mathcal{Y}$,
   where $\|\mathbf{Y}\|_{\infty, 2}$ denotes the maximum of the $\ell_{2}$-norm of the columns of~$\mathbf{Y}$.

   Let $\widetilde{\X}_{\star}, \widetilde{\mathbf{Y}}_{\star}$ be the optimal pair on $\widetilde{\mathbf{A}}$, \textit{i.e.},
   \begin{align}
      \bigl( \widetilde{\X}_{\star}, \widetilde{\mathbf{Y}}_{\star} \bigr) \eqdef \argmax_{\mathbf{X}\in \mathcal{X}, \mathbf{Y}\in \mathcal{Y}}\trace\bigl(\X^{\transpose}{\widetilde{\mathbf{A}}}\mathbf{Y}\bigr)
      \nonumber
   \end{align}
   and define the $r \times k$ matrix $\widetilde{\mathbf{C}}_{\star} \eqdef \widetilde{\mathbf{V}}^{\transpose}\widetilde{\mathbf{Y}}_{\star}$.
   Note that
   \begin{align}
      \|\widetilde{\mathbf{C}}_{\star}\|_{\infty, 2}
      &=
      \|\widetilde{\mathbf{V}}^{\transpose}\widetilde{\mathbf{Y}}_{\star}\|_{\infty, 2}
      \nonumber\\&=
      \max_{1\le i \le k}
      \|\widetilde{\mathbf{V}}^{\transpose}[\widetilde{\mathbf{Y}}_{\star}]_{:,i}\|_{2}
      \le 1,
      \label{norm-of-cstar}
   \end{align}
   with the last inequality following from the facts that
   $\|\mathbf{Y}\|_{\infty, 2} \le 1$ $\forall\; \mathbf{Y}\in \mathcal{Y}$
   and the columns of~$\widetilde{\mathbf{V}}$ are orthonormal.
   Alg.~\ref{algo:bilinear-lowrank} iterates over the points in  $(\mathbb{B}_{2}^{r-1})^{\otimes k}$.
   The latter is used to describe the set of $r \times k$ matrices whose columns have $\ell_{2}$ norm at most equal to $1$.
   At each point, the algorithm computes a candidate solution.
  By~\eqref{norm-of-cstar}, 
the $\epsilon$-net contains an $r \times k$ matrix $\mathbf{C}_{\sharp}$ such that
   \begin{align}
      \|\mathbf{C}_{\sharp} - \widetilde{\mathbf{C}}_{\star}\|_{\infty, 2} \le \epsilon.
      \nonumber
   \end{align}
   Let $\mathbf{X}_{\sharp}, \mathbf{Y}_{\sharp}$ be the candidate pair computed at $\mathbf{C}_{\sharp}$ by the two step maximization, \textit{i.e.},
   \begin{align}
      \mathbf{X}_{\sharp} \eqdef
      \argmax_{\mathbf{X}\in \mathcal{X}}
      \trace\bigl(
         \X^{\transpose}
         \widetilde{\mathbf{U}}\widetilde{\mathbf{\Sigma}}
         \mathbf{C}_{\sharp}
      \bigr) \nonumber
\intertext{and}
      \mathbf{Y}_{\sharp} \eqdef
      \argmax_{\mathbf{Y}\in \mathcal{Y}}
      \trace\bigl(
         \mathbf{X}_{\sharp}^{\transpose}
         \widetilde{\mathbf{A}} 
         \mathbf{Y}
      \bigr).
      \label{def-of-nearXandY}
   \end{align}
   We show that the objective values achieved by the candidate pair $\mathbf{X}_{\sharp}, \mathbf{Y}_{\sharp}$ satisfies the inequality of the lemma implying the desired result.

   By the definition of $\widetilde{\mathbf{C}}_{\star}$ and the linearity of the trace,
   \begin{align}
      &\trace\bigl(
         \widetilde{\X}_{\star}^{\transpose}
         \widetilde{\mathbf{A}}
         \widetilde{\mathbf{Y}}_{\star}
      \bigr)
      \nonumber\\&=
      \trace\bigl(
         \widetilde{\X}_{\star}^{\transpose}
         \widetilde{\mathbf{U}}\widetilde{\mathbf{\Sigma}}
         \widetilde{\mathbf{C}}_{\star}
      \bigr)
      \nonumber\\&=
      \trace\bigl(
         \widetilde{\X}_{\star}^{\transpose}
         \widetilde{\mathbf{U}}\widetilde{\mathbf{\Sigma}}
         \mathbf{C}_{\sharp}
      \bigr)
      +
      \trace\bigl(
         \widetilde{\X}_{\star}^{\transpose}
         \widetilde{\mathbf{U}}\widetilde{\mathbf{\Sigma}}
         \bigl(\widetilde{\mathbf{C}}_{\star}-\mathbf{C}_{\sharp}\bigr)
      \bigr)
      \nonumber\\&\le
      \trace\bigl(
         \mathbf{X}_{\sharp}^{\transpose}
         \widetilde{\mathbf{U}}\widetilde{\mathbf{\Sigma}}
         \mathbf{C}_{\sharp}
      \bigr)
      +
      \trace\bigl(
         \widetilde{\X}_{\star}^{\transpose}
         \widetilde{\mathbf{U}}\widetilde{\mathbf{\Sigma}}
         \bigl(\widetilde{\mathbf{C}}_{\star}-\mathbf{C}_{\sharp}\bigr)
      \bigr).
      \label{base-inequality}
   \end{align}
   The inequality follows from the fact that (by definition~\eqref{def-of-nearXandY}) $\mathbf{X}_{\sharp}$ maximizes the first term over all $\mathbf{X}\in \mathcal{X}$.
   We compute an upper bound on the right hand side of~\eqref{base-inequality}.
   Define
   $$
      \Yhat \eqdef \argmin_{\mathbf{Y}\in \mathcal{Y}}\|\widetilde{V}^{\transpose}\mathbf{Y}- \mathbf{C}_{\sharp}\|_{\infty, 2}.
   $$
   (We note that $\Yhat$ is used for the analysis and is never explicitly calculated.) 
   Further, define the $r \times k$ matrix $\widehat{\mathbf{C}} \eqdef \widetilde{\mathbf{V}}^{\transpose}\widehat{\mathbf{Y}}$.
   By the linearity of the trace operator
   \begin{align}
  &    \trace\bigl(
         \mathbf{X}_{\sharp}^{\transpose}
         \widetilde{\mathbf{U}}\widetilde{\mathbf{\Sigma}}
         \mathbf{C}_{\sharp}
      \bigr)\nonumber\\
   \nonumber&=
      \trace\bigl(
         \mathbf{X}_{\sharp}^{\transpose}
         \widetilde{\mathbf{U}}\widetilde{\mathbf{\Sigma}}
         \widehat{\mathbf{C}}
      \bigr)
      +
      \trace\bigl(
         \mathbf{X}_{\sharp}^{\transpose}
         \widetilde{\mathbf{U}}\widetilde{\mathbf{\Sigma}}
         \bigl(\mathbf{C}_{\sharp}-\widehat{\mathbf{C}}\bigr)
      \bigr)
   \nonumber\\&=
      \trace\bigl(
         \mathbf{X}_{\sharp}^{\transpose}
         \widetilde{\mathbf{U}}\widetilde{\mathbf{\Sigma}} \widetilde{\mathbf{V}}^{\transpose}
         \Yhat
      \bigr)
      +
      \trace\bigl(
         \mathbf{X}_{\sharp}^{\transpose}
         \widetilde{\mathbf{U}}\widetilde{\mathbf{\Sigma}}
         \bigl(\mathbf{C}_{\sharp}-\widehat{\mathbf{C}}\bigr)
      \bigr)
   \nonumber\\&\le
      \trace\bigl(
         \mathbf{X}_{\sharp}^{\transpose}
         \widetilde{\mathbf{U}}\widetilde{\mathbf{\Sigma}} \widetilde{\mathbf{V}}^{\transpose}
         \mathbf{Y}_{\sharp}
      \bigr)
      +
      \trace\bigl(
         \mathbf{X}_{\sharp}^{\transpose}
         \widetilde{\mathbf{U}}\widetilde{\mathbf{\Sigma}}
         \bigl(\mathbf{C}_{\sharp}-\widehat{\mathbf{C}}\bigr)
      \bigr)
      \nonumber\\&=
      \trace\bigl(
         \mathbf{X}_{\sharp}^{\transpose}
         \widetilde{\mathbf{A}}
         \mathbf{Y}_{\sharp}
      \bigr)
      +
      \trace\bigl(
         \mathbf{X}_{\sharp}^{\transpose}
         \widetilde{\mathbf{U}}\widetilde{\mathbf{\Sigma}}
         \bigl(\mathbf{C}_{\sharp}-\widehat{\mathbf{C}}\bigr)
      \bigr).
      \label{asdfsadfsafd}
   \end{align}
   The inequality follows from the fact that (by definition~\eqref{def-of-nearXandY}) $\mathbf{Y}_{\sharp}$ maximizes the first term over all~$\mathbf{Y}\in \mathcal{Y}$.
   Combining~\eqref{asdfsadfsafd} and~\eqref{base-inequality}, and rearranging the terms, we obtain
   \begin{align}
      &\trace\bigl(
         \widetilde{\X}_{\star}^{\transpose}
         \widetilde{\mathbf{A}}
         \widetilde{\mathbf{Y}}_{\star}
      \bigr)
      -\trace\bigl(
         \mathbf{X}_{\sharp}^{\transpose}
         \widetilde{\mathbf{A}}
         \mathbf{Y}_{\sharp}
      \bigr) \nonumber\\
      &\le
      \trace\bigl(
         \widetilde{\X}_{\star}^{\transpose}
         \widetilde{\mathbf{U}}\widetilde{\mathbf{\Sigma}}
         \bigl(\widetilde{\mathbf{C}}_{\star}-\mathbf{C}_{\sharp}\bigr)
      \bigr)
      +
      \trace\bigl(
         \mathbf{X}_{\sharp}^{\transpose}
         \widetilde{\mathbf{U}}\widetilde{\mathbf{\Sigma}}
         \bigl(\mathbf{C}_{\sharp}-\widehat{\mathbf{C}}\bigr)
      \bigr).
      \label{new-base-inequality}
   \end{align}
   By Lemma~\ref{lemma:abs-trace-XAY-ub},
   \begin{align}
      &\bigl\lvert
      \trace\bigl(
         \widetilde{\X}_{\star}^{\transpose}
         \widetilde{\mathbf{U}}\widetilde{\mathbf{\Sigma}}
         \bigl(\widetilde{\mathbf{C}}_{\star}-\mathbf{C}_{\sharp}\bigr)
      \bigr)
      \bigr\rvert
   \nonumber\\&\qquad\le
      \|\widetilde{\X}_{\star}^{\transpose}\widetilde{\mathbf{U}}\|_{\frob}
      \cdot \|\widetilde{\mathbf{\Sigma}}\|_{2}
      \cdot \|\widetilde{\mathbf{C}}_{\star}-\mathbf{C}_{\sharp}\|_{\frob}
   \nonumber\\&\qquad\le
      \|\widetilde{\X}_{\star}\|_{\frob}
      \cdot \sigma_{1}(\widetilde{\mathbf{A}})
      \cdot \sqrt{k} \cdot \epsilon
   \nonumber\\&\qquad\le
      \max_{\mathbf{X}\in \mathcal{X}}\|\mathbf{X}\|_{\frob}
      \cdot \sigma_{1}(\widetilde{\mathbf{A}})
      \cdot \sqrt{k} \cdot \epsilon
   \nonumber\\&\qquad\le
      \sigma_{1}(\widetilde{\mathbf{A}}) \cdot \sqrt{k} \cdot \epsilon.
      \label{first-abs-trace}
   \end{align}
   Similarly,
   \begin{align}
   \bigl\lvert
      \trace\bigl(
         \mathbf{X}_{\sharp}^{\transpose}
         \widetilde{\mathbf{U}}\widetilde{\mathbf{\Sigma}}
         \bigl(\mathbf{C}_{\sharp}-\widehat{\mathbf{C}}\bigr)
      \bigr)
      \bigr\rvert
   &\le
      \|\mathbf{X}_{\sharp}\widetilde{\mathbf{U}}\|_{\frob} \cdot \|\widetilde{\mathbf{\Sigma}}\|_{2} \cdot \|\mathbf{C}_{\sharp}-\widehat{\mathbf{C}}\|_{\frob}
   \nonumber\\&\le
      \max_{\mathbf{X}\in \mathcal{X}}\|\mathbf{X}\|_{\frob}
      \cdot \sigma_{1}(\widetilde{\mathbf{A}})
      \cdot \sqrt{k} \cdot \epsilon
   \nonumber\\&\le
      \sigma_{1}(\widetilde{\mathbf{A}})
      \cdot \sqrt{k} \cdot \epsilon.
      \label{second-abs-trace}
   \end{align}
   The second inequality follows from the fact that by the definition of $\widehat{\mathbf{C}}$,
   \begin{align}
      \|\widehat{\mathbf{C}} - \mathbf{C}_{\sharp}\|_{\infty, 2}
      &=
      \|\widetilde{V}^{\transpose}\Yhat - \mathbf{C}_{\sharp}\|_{\infty, 2}
      \le
      \|\widetilde{V}^{\transpose}\widetilde{\mathbf{Y}}_{\star} - \mathbf{C}_{\sharp}\|_{\infty, 2}
      \nonumber\\&=
      \|\widetilde{\mathbf{C}}_{\star} - \mathbf{C}_{\sharp}\|_{\infty, 2}
      \le
      \epsilon,
      \nonumber
   \end{align}
   which implies that
   \begin{align}
      \|\widehat{\mathbf{C}} - \mathbf{C}_{\sharp}\|_{\frob}
      \le
      \sqrt{k} \cdot \epsilon.
      \nonumber
   \end{align}
   Continuing from~\eqref{new-base-inequality}
   under~\eqref{first-abs-trace} and~\eqref{second-abs-trace},
   \begin{align}
      \trace\bigl(
         \mathbf{X}_{\sharp}^{\transpose}
         \widetilde{\mathbf{A}}
         \mathbf{Y}_{\sharp}
      \bigr)
      \ge
      \trace\bigl(
         \widetilde{\X}_{\star}^{\transpose}
         \widetilde{\mathbf{A}}
         \widetilde{\mathbf{Y}}_{\star}
      \bigr)
      -
      2\cdot \epsilon
      \cdot \sqrt{k}
      \cdot \sigma_{1}(\widetilde{\mathbf{A}}).
      \nonumber
   \end{align}
   Recalling that the singular values of $\widetilde{\mathbf{A}}$
   have been scaled by a factor of $\mu_{\mathcal{X}} \cdot \mu_{\mathcal{Y}}$ yields the desired result.
   
   The runtime of Alg.~\ref{algo:bilinear-lowrank} follows from the cost per iteration and the cardinality of the $\epsilon$-net.
   Matrix multiplications can exploit the truncated singular value decomposition of $\widetilde{\mathbf{A}}$ which is performed only once.
\end{proof}

\begin{lemma}
   \label{lemma:generic-sketch-solution}
   For any $\mathbf{A}, \widetilde{\mathbf{A}} \in \mathbb{R}^{m \times n}$, 
   and norm-boudned sets $\mathcal{X} \subseteq \mathbb{R}^{m \times k}$ 
   and $\mathcal{Y} \subseteq \mathbb{R}^{n \times k}$,
   let 
   $$
   \bigl( \mathbf{X}_{\star}, \Y_{\star} \bigr) 
   \eqdef 
   \argmax_{\mathbf{X}\in \mathcal{X}, \mathbf{Y}\in \mathcal{Y}}\trace\bigl(\X^{\transpose}\mathbf{A}\mathbf{Y}\bigr),
   $$
   and 
   $$
   \bigl( \widetilde{\X}_{\star}, \widetilde{\mathbf{Y}}_{\star} \bigr)
   \eqdef
   \argmax_{\mathbf{X}\in \mathcal{X}, \mathbf{Y}\in \mathcal{Y}}\trace\bigl(\X^{\transpose}{\widetilde{\mathbf{A}}}\mathbf{Y}\bigr).
   $$
   For any $(\widetilde{\X}, \widetilde{\mathbf{Y}}) \in \mathcal{X} \times \mathcal{Y}$ such that
   $$
   \trace\bigl(\widetilde{\X}^{\transpose}{\widetilde{\mathbf{A}}}\widetilde{\mathbf{Y}}\bigr)
   \ge \gamma \cdot 
   \trace\bigl(\widetilde{\X}_{\star}^{\transpose}{\widetilde{\mathbf{A}}}\widetilde{\mathbf{Y}}_{\star}\bigr) - C
   $$
   for some $0 < \gamma \le 1$, we have
   \begin{align}
	  \trace\bigl(\widetilde{\X}^{\transpose}\mathbf{A}\widetilde{\mathbf{Y}}\bigr)
	  &\ge
	  \gamma \cdot \trace\bigl(\mathbf{X}_{\star}^{\transpose}\mathbf{A}\Y_{\star}\bigr)
	  - C
	  \nonumber\\&\quad
	  - 2 \cdot \|\mathbf{A}-{\widetilde{\mathbf{A}}}\|_{2} 
     \cdot \mu_{\mathcal{X}}
	  \cdot
     \mu_{\mathcal{Y}}.
	  \nonumber
   \end{align}
   where
   $\mu_{\mathcal{X}}\eqdef \max_{\mathbf{X}\in  \mathcal{X}}\| \mathbf{X}\|_{\frob}$
   and
   $\mu_{\mathcal{Y}}\eqdef \max_{\mathbf{Y}\in  \mathcal{Y}}\| \mathbf{Y}\|_{\frob}$.
\end{lemma}
\begin{proof}
By the optimality of 
$\widetilde{\X}_{\star}, \widetilde{\mathbf{Y}}_{\star}$ for ${\widetilde{\mathbf{A}}}$, we have
\begin{align}
   \trace\bigl(\widetilde{\X}_{\star}^{\transpose}{\widetilde{\mathbf{A}}}\widetilde{\mathbf{Y}}_{\star}\bigr)
   \ge
   \trace\bigl(\mathbf{X}_{\star}^{\transpose}{\widetilde{\mathbf{A}}}\Y_{\star}\bigr).
   \nonumber
\end{align}
In turn,
for any $(\widetilde{\X}, \widetilde{\mathbf{Y}}) \in \mathcal{X} \times \mathcal{Y}$ such that
 $$
   \trace\bigl(\widetilde{\X}^{\transpose}{\widetilde{\mathbf{A}}}\widetilde{\mathbf{Y}}\bigr)
   \ge \gamma \cdot 
   \trace\bigl(\widetilde{\X}_{\star}^{\transpose}{\widetilde{\mathbf{A}}}\widetilde{\mathbf{Y}}_{\star}\bigr) - C
 $$

 for some $0 < \gamma < 1$ (if such pairs exist), we have
\begin{align}
   \trace\bigl(\widetilde{\X}^{\transpose}{\widetilde{\mathbf{A}}}\widetilde{\mathbf{Y}}\bigr)
   \ge
   \gamma \cdot
   \trace\bigl(\mathbf{X}_{\star}^{\transpose}{\widetilde{\mathbf{A}}}\Y_{\star}\bigr) - C.
   \label{eq:bp01}
\end{align}
By the linearity of the trace operator,
\begin{align}
   &\trace\bigl(\widetilde{\X}^{\transpose}{\widetilde{\mathbf{A}}}\widetilde{\mathbf{Y}}\bigr)
   \nonumber\\
   &\qquad =
   \trace\bigl(\widetilde{\X}^{\transpose}\mathbf{A}\widetilde{\mathbf{Y}}\bigr)
   -
   \trace\bigl(\widetilde{\X}^{\transpose}(\mathbf{A}-{\widetilde{\mathbf{A}}})\widetilde{\mathbf{Y}}\bigr)
   \nonumber\\ 
   &\qquad \le
   \trace\bigl(\widetilde{\X}^{\transpose}\mathbf{A}\widetilde{\mathbf{Y}}\bigr)
   +
   \bigl\lvert\trace\bigl(\widetilde{\X}^{\transpose}(\mathbf{A}-{\widetilde{\mathbf{A}}})\widetilde{\mathbf{Y}}\bigr)\bigr\rvert.
   \label{eq:bp02}
\end{align}
By Lemma~\ref{lemma:abs-trace-XAY-ub},
\begin{align}
   &\bigl\lvert\trace\bigl(\widetilde{\X}^{\transpose}(\mathbf{A}-{\widetilde{\mathbf{A}}})\widetilde{\mathbf{Y}}\bigr)\bigr\rvert
   \nonumber\\
   &\le
   \|\widetilde{\X}\|_{\frob} \cdot \|\widetilde{\mathbf{Y}}\|_{\frob} \cdot \|\mathbf{A}-{\widetilde{\mathbf{A}}}\|_{2}
   \nonumber\\
   &\le 
   \|\mathbf{A}-{\widetilde{\mathbf{A}}}\|_{2}
   \cdot \max_{\mathbf{X}\in \mathcal{X}} \|\mathbf{X}\|_{\frob}
   \cdot \max_{\mathbf{Y}\in \mathcal{Y}} \|\mathbf{Y}\|_{\frob}
   \;\eqdef\; R.
   \label{residual-trace-bound}
\end{align}
Continuing from~\eqref{eq:bp02}, 
\begin{align}
   \trace\bigl(\widetilde{\X}^{\transpose}{\widetilde{\mathbf{A}}}\widetilde{\mathbf{Y}}\bigr)
   \le
   \trace\bigl(\widetilde{\X}^{\transpose}\mathbf{A}\widetilde{\mathbf{Y}}\bigr)
   +
   R.
   \label{eq:bp04}
\end{align}
Similarly,
\begin{align}
   &\trace\bigl(\mathbf{X}_{\star}^{\transpose}{\widetilde{\mathbf{A}}}\Y_{\star}\bigr)
   \nonumber\\
   &=
   \trace\bigl(\mathbf{X}_{\star}^{\transpose}\mathbf{A}\Y_{\star}\bigr)
   -
   \trace\bigl(\mathbf{X}_{\star}^{\transpose}(\mathbf{A}-{\widetilde{\mathbf{A}}})\Y_{\star}\bigr)
   \nonumber\\& \ge
   \trace\bigl(\mathbf{X}_{\star}^{\transpose}\mathbf{A}\Y_{\star}\bigr)
   -
   \bigl\lvert
   \trace\bigl(\mathbf{X}_{\star}^{\transpose}(\mathbf{A}-{\widetilde{\mathbf{A}}})\Y_{\star}\bigr)
   \bigr\rvert
   \nonumber\\& \ge
   \trace\bigl(\mathbf{X}_{\star}^{\transpose}\mathbf{A}\Y_{\star}\bigr)
   -   
   R.
   \label{eq:bp03}
\end{align}
Combining the above, we have
\begin{align}
   \trace\bigl(\widetilde{\X}^{\transpose}\mathbf{A}\widetilde{\mathbf{Y}}\bigr)
   &\ge
   \trace\bigl(\widetilde{\X}^{\transpose}{\widetilde{\mathbf{A}}}\widetilde{\mathbf{Y}}\bigr)
   -R
   \nonumber\\&\ge
   \gamma \cdot
   \trace\bigl(\mathbf{X}_{\star}^{\transpose}{\widetilde{\mathbf{A}}}\Y_{\star}\bigr)
   -
   R -C
   \nonumber\\&\ge
   \gamma \cdot
   \bigl(
   \trace\bigl(\mathbf{X}_{\star}^{\transpose}\mathbf{A}\Y_{\star}\bigr)
   -
   R
   \bigr) - R -C
   \nonumber\\&=
   \gamma \cdot \trace\bigl(\mathbf{X}_{\star}^{\transpose}\mathbf{A}\Y_{\star}\bigr)
   - (1 + \gamma) \cdot  R -C
   \nonumber\\&\ge
   \gamma \cdot \trace\bigl(\mathbf{X}_{\star}^{\transpose}\mathbf{A}\Y_{\star}\bigr)
   - 2 \cdot  R -C,
   \nonumber
\end{align}
where 
the first inequality follows from~\eqref{eq:bp04} 
the second from~\eqref{eq:bp01},
the third from~\eqref{eq:bp03}, and the last from the fact that $R \ge 0$.
This concludes the proof.
\end{proof}
\begin{replemma}{full-rank-guarantees}
   For any $\mathbf{A} \in \mathbb{R}^{m \times n}$, let
   $$
   \bigl( \mathbf{X}_{\star}, \Y_{\star} \bigr) 
   \eqdef 
   \argmax_{\mathbf{X}\in \mathcal{X}, \mathbf{Y}\in \mathcal{Y}}\trace\bigl(\X^{\transpose}\mathbf{A}\mathbf{Y}\bigr),
   $$
   where $\mathcal{X} \subseteq \mathbb{R}^{m \times k}$ and $\mathcal{Y} \subseteq \mathbb{R}^{n \times k}$
   are sets satisfying the conditions of Lemma~\ref{low-rank-solver-guarantees}.
   Let $\widetilde{\mathbf{A}}$ be a rank-$r$ approximation of $\mathbf{A}$,
   and $\widetilde{\X} \in \mathcal{X}$, $\widetilde{\mathbf{Y}} \in \mathcal{Y}$ 
   be the output of Alg.~\ref{algo:bilinear-lowrank}
   with input $\widetilde{\mathbf{A}}$ and accuracy $\epsilon$. Then,
   \begin{align}
      &\trace\bigl(\mathbf{X}_{\star}^{\transpose}\mathbf{A}\Y_{\star}\bigr)
      -
      \trace\bigl(\widetilde{\X}^{\transpose}\mathbf{A}\widetilde{\mathbf{Y}}\bigr)
      \nonumber\\&\quad \le
      2  \cdot \left(
         \epsilon  \sqrt{k} \cdot \|\widetilde{\mathbf{A}}\|_{2}
         +
         \|\mathbf{A}-{\widetilde{\mathbf{A}}}\|_{2} 
         \right)
     \cdot \mu_{\mathcal{X}}
     \cdot \mu_{\mathcal{Y}},
     \nonumber
   \end{align}
   where
   $\mu_{\mathcal{X}}\eqdef \max_{\mathbf{X}\in  \mathcal{X}}\| \mathbf{X}\|_{\frob}$
   and
   $\mu_{\mathcal{Y}}\eqdef \max_{\mathbf{Y}\in  \mathcal{Y}}\| \mathbf{Y}\|_{\frob}$.
\end{replemma}
\begin{proof}
The proof follows the approximation guarantees of Alg.~\ref{algo:bilinear-lowrank} in Lemma~\ref{low-rank-solver-guarantees}
and Lemma~\ref{lemma:generic-sketch-solution}.
\end{proof}

\section{Correctness of Algorithm~\ref{algo:zero-one-orthogonal}}
\label{sec:proof-correctness-of-local-solver}
In the sequel, we use $\|\mathbf{X}\|_{\infty, 1}$ to denote the maximum of the $\ell_{1}$ norm of the rows of $\mathbf{X}$.
When $\mathbf{X} \in \lbrace 0, 1 \rbrace^{d\times k}$, the constraint
$\|\mathbf{X}\|_{\infty, 1}=1$ effectively implies that each row of $\X$ has exactly one nonzero entry.
\begin{replemma}{lemma:local-solver-for-bcc}
   Let 
   $
   \mathcal{X}
   \eqdef
   \bigl\lbrace \mathbf{X}\in \lbrace 0, 1\rbrace^{d \times k}: \|\mathbf{X}\|_{\infty, 1} = 1\bigr\rbrace.
   $
   For any $d \times k$ real matrix $\mathbf{L}$, Algorithm~\ref{algo:zero-one-orthogonal}
 outputs
 $$
   \widetilde{\mathbf{X}}
   =
   \argmax_{\mathbf{X}\in \mathcal{X}}\trace\bigl(\X^{\transpose}\mathbf{L}\bigr),
 $$
 in time $O({k}\cdot{d})$
\end{replemma}

\begin{proof}
   By construction, each row of $\X$ has exactly one nonzero entry.
   Let $j_{i} \in [k]$ denote the index of the nonzero entry in the $i$th row of $\X$.
   For any ${\mathbf{X}\in \mathcal{X}}$,
   \begin{align}
     \trace\bigl(\X^{\transpose}\mathbf{L}\bigr)
     &=
     \sum_{j=1}^{k} \mathbf{x}_{j}^{\transpose}\mathbf{l}_{j}
     =
     \sum_{j=1}^{k} \sum_{i \in \supp(\mathbf{x}_{j})} 1\cdot L_{ij} \nonumber \\
     &=
     \sum_{i=1}^{d} L_{ij_{i}}
     \le
     \sum_{i=1}^{d} \max_{j\in [k]} L_{ij}.
     \label{binary-objective-upperbound}
   \end{align}
   Algorithm~\ref{algo:zero-one-orthogonal} achieves equality in~\eqref{binary-objective-upperbound} due to the choice of $j_{i}$ in line 3.
   Finally, the running time follows immediately from the $O(k)$ time required to determine the maximum entry of each of the $d$ rows of~$\mathbf{L}$.
\end{proof}

\section{Auxiliary Lemmas}
\begin{lemma}
   \label{holder-consequence}
   Let $a_1,\hdots , a_n$ 
   and $b_1, \hdots, b_n$ be $2n$ real numbers and let $p$ and
   $q$ be two numbers such that ${1/p} + {1/q} = 1$ and $p>1$. We have
   \begin{align}
	  \left\lvert 
		 \sum_{i=1}^{n} a_{i}b_{i}
	  \right\rvert
	  \le
	  \left( \sum_{i=1}^{n} \lvert  a_{i}\rvert^{p} \right)^{1/p}
	  \cdot
	  \left( \sum_{i=1}^{n} \lvert  b_{i}\rvert^{q} \right)^{1/q}.
	  \nonumber
   \end{align}
\end{lemma}
\begin{lemma}
   \label{lemma:inner-prod-ub}
   For any $\mathbf{A}, \mathbf{B} \in \mathbb{R}^{n \times k}$,
   \begin{align}
	  \bigl\lvert 
		 \langle \mathbf{A}, \mathbf{B} \rangle
	  \bigr\rvert 
	  \eqdef 
	  \bigl\lvert
		 \trace\bigl( \mathbf{A}^{\transpose} \mathbf{B}\bigr)
	  \bigr\rvert
	  \le
	  \|\mathbf{A}\|_{\frob}
	  \|\mathbf{B}\|_{\frob}.
	  \nonumber
   \end{align}
\end{lemma}
\begin{proof}
 Treating $\mathbf{A}$ and $\mathbf{B}$ as vectors, 
 the lemma follows immediately from Lemma~\ref{holder-consequence} for $p=q=2$.
\end{proof}
\begin{lemma}
   \label{lemma:frob-of-matrix-prod}
   For any two real matrices $\mathbf{A}$ and $\mathbf{B}$ of appropriate dimensions,
   \begin{align}
	  \|\mathbf{A}\mathbf{B}\|_{\frob}
	  \le
	  \min\mathopen{}
	  \bigl\lbrace
		 \|\mathbf{A}\|_{2} \|\mathbf{B}\|_{\frob}, \;
		 \|\mathbf{A}\|_{\frob} \|\mathbf{B}\|_{2}
	  \bigr\rbrace.
	  \nonumber
   \end{align}
\end{lemma}
\begin{proof}
   Let $\mathbf{b}_{i}$ denote the $i$th column of $\mathbf{B}$. 
   Then,
   \begin{align}
	  \|\mathbf{A}\mathbf{B}\|_{\frob}^{2}
	 & =
	  \sum_{i} \| \mathbf{A} \mathbf{b}_{i} \|_{2}^{2}
	  \le
	  \sum_{i} \| \mathbf{A}\|_{2}^{2}  \|\mathbf{b}_{i} \|_{2}^{2}
	  \nonumber\\&=
	  \| \mathbf{A}\|_{2}^{2} \sum_{i}   \|\mathbf{b}_{i} \|_{2}^{2}
	  =
	  \| \mathbf{A}\|_{2}^{2} \|\mathbf{B} \|_{\frob}^{2}.
	  \nonumber 
   \end{align}
   Similarly, using the previous inequality,
   \begin{align}
	  \|\mathbf{A}\mathbf{B}\|_{\frob}^{2}
	  =
	  \|\mathbf{B}^{\transpose}\mathbf{A}^{\transpose}\|_{\frob}^{2}
	  \le
	  \|\mathbf{B}^{\transpose}\|_{2}^{2}\|\mathbf{A}^{\transpose}\|_{\frob}^{2}
	  =
	  \|\mathbf{B}\|_{2}^{2}\|\mathbf{A}\|_{\frob}^{2}.
	  \nonumber
   \end{align}
   The desired result follows combining the two upper bounds.
\end{proof}
\begin{lemma}
   \label{lemma:abs-trace-XAY-ub}
   For any real 
   $m \times k$ matrix $\X$, $m \times n$ matrix $\mathbf{A}$, and  $n \times k$ matrix $\Y$,
   \begin{align}
	  \bigl\lvert
	  \trace\bigl( \X^{\transpose} \mathbf{A} \mathbf{Y}\bigr)
	  \bigr\rvert
	  \le
	  \|\mathbf{X}\|_{\frob} \cdot \|\mathbf{A} \|_{2} \cdot \|\mathbf{Y}\|_{\frob}.
	  \nonumber
   \end{align}
\end{lemma}
\begin{proof}
   We have
   \begin{align}
	  \bigl\lvert
	  \trace\bigl( \X^{\transpose} \mathbf{A} \mathbf{Y}\bigr)
	  \bigr\rvert
	  \le
	  \|\mathbf{X}\|_{\frob} \cdot \|\mathbf{A}\mathbf{Y}\|_{\frob}
	  \le
	  \|\mathbf{X}\|_{\frob} \cdot \|\mathbf{A} \|_{2} \cdot \|\mathbf{Y}\|_{\frob},
	  \nonumber
   \end{align}
   with the first inequality following from Lemma~\ref{lemma:inner-prod-ub}
   on $\lvert \langle \X,\, \mathbf{A}\mathbf{Y}\rangle\rvert$ and the second from Lemma~\ref{lemma:frob-of-matrix-prod}.
\end{proof}
\begin{lemma}
   \label{lemma:abs-trace-XAY-ub-orthogonal}
   For any real
   $m \times n$ matrix $\mathbf{A}$,
   and pair of 
   $m \times k$ matrix $\X$ and  $n \times k$ matrix $\Y$
   such that $\X^{\transpose}\X=\mathbf{I}_{k}$ and $\Y^{\transpose}\Y=\mathbf{I}_{k}$
   with $k \le \min\lbrace {m},\; {n}\rbrace$, the following holds:
   \begin{align}
	  \bigl\lvert
	  \trace\bigl( \X^{\transpose} \mathbf{A} \mathbf{Y}\bigr)
	  \bigr\rvert
	  \le
	  \sqrt{k}\cdot 
	  \bigl( \sum_{i=1}^{k} \sigma_{i}^{2}\bigl(\mathbf{A}\bigr) \bigr)^{1/2}.
	  \nonumber
   \end{align}
\end{lemma}
\begin{proof}
   By Lemma~\ref{lemma:inner-prod-ub},
   \begin{align}
	  \lvert \langle \X,\, \mathbf{A}\mathbf{Y}\rangle\rvert
	  &=
	  \bigl\lvert
	  \trace\bigl( \X^{\transpose} \mathbf{A} \mathbf{Y}\bigr)
	  \bigr\rvert
	  \nonumber\\&
	  \le
	  \|\mathbf{X}\|_{\frob} \cdot \|\mathbf{A}\mathbf{Y}\|_{\frob}
	  =
	  \sqrt{k} \cdot \|\mathbf{A}\mathbf{Y}\|_{\frob}.
	  \nonumber
   \end{align}
   where the last inequality follows from the fact that $\|\mathbf{X}\|_{\frob}^{2} = \trace\bigl(\X^{\transpose}\mathbf{X}\bigr) = \trace\bigl(\mathbf{I}_{k}\bigr) = k$.
   Further, for any $\Y$ such that $\Y^{T}\mathbf{Y}= \mathbf{I}_{k}$,
   \begin{align}
   		\|\mathbf{A} \mathbf{Y}\|_{\frob}^{2}
		\le
   		\max_{
			\substack{
				\widehat{\mathbf{Y}} \in \mathbb{R}^{n \times k}\\
				\widehat{\mathbf{Y}}^{\transpose}\widehat{\mathbf{Y}} = \mathbf{I}_{k} 
			}
		}
   		\|\mathbf{A} \widehat{\mathbf{Y}} \|_{\frob}^{2}
		=
		\sum_{i=1}^{k} \sigma_{i}^{2}(\mathbf{A}).
   \end{align}
   Combining the two inequalities, the result follows. 
\end{proof}
\begin{lemma}
   For any real
   $m \times n$ matrix $\mathbf{A}$,
   and any $k \le \min\lbrace {m},\; {n}\rbrace$, 
      \begin{align}
   		\max_{
			\substack{
				{\mathbf{Y}} \in \mathbb{R}^{n \times k}\\
				{\mathbf{Y}}^{\transpose}{\mathbf{Y}} = \mathbf{I}_{k} 
			}
		}
   		\|\mathbf{A} {\mathbf{Y}} \|_{\frob}
		=
		\left(\sum_{i=1}^{k} \sigma_{i}^{2}(\mathbf{A}) \right)^{1/2}.
		\nonumber
   \end{align}
   The above equality is realized when the $k$ columns of~$\mathbf{Y}$ coincide with the $k$ leading right singular vectors of~$\mathbf{A}$.
\end{lemma}
\begin{proof}
	Let $\mathbf{U}\mathbf{\Sigma}\mathbf{V}^{\transpose}$ be the singular value decomposition of $\mathbf{A}$,
   with
   $\Sigma_{jj} = \sigma_{j}$ being the $j$th largest singular value of $\mathbf{A}$, $j=1, \hdots, d$, where $d \eqdef \min\lbrace {m}, {n} \rbrace$.
	Due to the invariance of the Frobenius norm under unitary multiplication,
   \begin{align}
   		\|\mathbf{A} \mathbf{Y}\|_{\frob}^{2}
		=
		\|\mathbf{U}\mathbf{\Sigma}\mathbf{V}^{\transpose} \mathbf{Y}\|_{\frob}^{2}
		=
		\|\mathbf{\Sigma}\mathbf{V}^{\transpose} \mathbf{Y}\|_{\frob}^{2}.
		\label{frob-norm-unitary}
   \end{align}
   Continuing from~\eqref{frob-norm-unitary},
   \begin{align}
		\|\mathbf{\Sigma}\mathbf{V}^{\transpose} \mathbf{Y}\|_{\frob}^{2}
		&=
		\trace\bigl(\mathbf{Y}^{\transpose}\mathbf{V}\mathbf{\Sigma}^{2}\mathbf{V}^{\transpose} \mathbf{Y}\bigr)
		\nonumber\\
		&=
		\sum_{i=1}^{k} \mathbf{y}_{i}^{\transpose} 
			\left( 
				\sum_{j=1}^{d} \sigma_{j}^{2} \cdot \mathbf{v}_{j} \mathbf{v}_{j}^{\transpose}
			\right)
			\mathbf{y}_{i}
		\nonumber\\
		&=
		\sum_{j=1}^{d} 
			\sigma_{j}^{2} \cdot \sum_{i=1}^{k}
			\left( \mathbf{v}_{j}^{\transpose} \mathbf{y}_{i}\right)^{2}.
		\nonumber
   \end{align}
   Let $z_{j} \eqdef \sum_{i=1}^{k} \bigl( \mathbf{v}_{j}^{\transpose} \mathbf{y}_{i}\bigr)^{2}$, $j=1, \hdots, d$.
   Note that each individual $z_{j}$ satisfies
   \begin{align}
   		 0 \le z_{j} \eqdef \sum_{i=1}^{k} \bigl( \mathbf{v}_{j}^{\transpose} \mathbf{y}_{i}\bigr)^{2}
		\le
		\|\mathbf{v}_{j}\|^{2} = 1,
		\nonumber
   \end{align}
   where the last inequality follows from the fact that the columns of $\mathbf{Y}$ are orthonormal.
   Further,
   \begin{align}
   		\sum_{j=1}^{d} z_{j} 
		&=  
		\sum_{j=1}^{d}\sum_{i=1}^{k} \bigl( \mathbf{v}_{j}^{\transpose} \mathbf{y}_{i}\bigr)^{2}
		=
		\sum_{i=1}^{k}\sum_{j=1}^{d} \bigl( \mathbf{v}_{j}^{\transpose} \mathbf{y}_{i}\bigr)^{2}
		\nonumber\\&=
		\sum_{i=1}^{k}\|\mathbf{y}_{i}\|^{2} = k.
		\nonumber
   \end{align}
   Combining the above, we conclude that
   \begin{align}
		\|\mathbf{A}\mathbf{Y}\|_{\frob}^{2}
		=
		\sum_{j=1}^{d} 
			\sigma_{j}^{2} \cdot z_{j}
		\le \sigma_{1}^{2} + \hdots + \sigma_{k}^{2}.
		\label{frob-prod-ub}
   \end{align}
   Finally, it is straightforward to verify that if $\mathbf{y}_{i} = \mathbf{v}_{i}$, $i=1, \hdots, k$, then~\eqref{frob-prod-ub} holds with equality.
\end{proof}
\begin{lemma}
   \label{lemma:partial-sing-val-sum}
   For any real $m \times n$ matrix $\mathbf{A}$, 
   let $\sigma_{i}(\mathbf{A})$ be the $i$th largest singular value.
   For any $r, k \le \min\lbrace m, n \rbrace$,
   \begin{align}
	  \sum_{i=r+1}^{r + k} \sigma_{i}(\mathbf{A})
	  \le
	  \frac{k}{\sqrt{r+k}} \|\mathbf{A}\|_{\frob}.
	  \nonumber
   \end{align} 
\end{lemma}
\begin{proof}
   By the Cauchy-Schwartz inequality,
   \begin{align}
	  \sum_{i=r+1}^{r + k} \sigma_{i}(\mathbf{A})
	  &=
	  \sum_{i=r+1}^{r + k} \lvert\sigma_{i}(\mathbf{A})\rvert
	  \le
	  \left( \sum_{i=r+1}^{r + k} \sigma_{i}^{2}(\mathbf{A}) \right)^{1/2}
	  \|\mathbf{1}_{k}\|_{2}
	  \nonumber\\&=
	  \sqrt{k} \cdot \left( \sum_{i=r+1}^{r + k} \sigma_{i}^{2}(\mathbf{A}) \right)^{1/2}.
	  \nonumber
   \end{align}
   Note that $\sigma_{r+1}(\mathbf{A}), \hdots, \sigma_{r+k}(\mathbf{A})$ are the $k$ smallest among the $r+k$ largest singular values. 
   Hence,
   \begin{align}
	  \sum_{i=r+1}^{r + k} \sigma_{i}^{2}(\mathbf{A})
	  &\le
	  \frac{k}{r+k}\sum_{i=1}^{r + k} \sigma_{i}^{2}(\mathbf{A})
	  \le
	  \frac{k}{r+k} \sum_{i=1}^{l} \sigma_{i}^{2}(\mathbf{A})
	  \nonumber\\&=
	  \frac{k}{r+k} \|\mathbf{A}\|_{\frob}^{2}.
	  \nonumber
   \end{align}
   Combining the two inequalities, the desired result follows.
\end{proof}
\begin{corollary}
	\label{sigma-bound}
   For any real $m \times n$ matrix $\mathbf{A}$,
   the $r$th largest singular value $\sigma_{\mathrm{r}}(\mathbf{A})$ satisfies
   $\sigma_{\mathrm{r}}(\mathbf{A}) \le  \|\mathbf{A}\|_{\frob} / \sqrt{r}$.
\end{corollary}
\begin{proof}
   It follows immediately from Lemma~\ref{lemma:partial-sing-val-sum}.
\end{proof}

First, we define the $\|\cdot\|_{\infty,2}$ norm of a matrix as the $l_{2}$ norm of the column with the maximum $l_{2}$ norm, \textit{i.e.}, for an $r \times k$ matrix $\mathbf{C}$
$$
   \|\mathbf{C}\|_{\infty, 2} = \max_{1 \le i \le k} \|\mathbf{c}_{i}\|_{2}.
$$
Note that
\begin{align}
   \|\mathbf{C}\|_{\frob}^{2}
   &=
   \sum_{i=1}^{k}\|\mathbf{c}_{i}\|_{2}^{2}
   \le
   k \cdot \max_{1 \le i \le k} \|\mathbf{c}_{i}\|_{2}^{2}
   \nonumber\\&=
   k \cdot \left( \max_{1 \le i \le k} \|\mathbf{c}_{i}\|_{2} \right)^{2}
   =
   k \cdot \|\mathbf{C}\|_{\infty, 2}.
\end{align}

\clearpage\newpage

\end{document}